\documentclass[lettersize,journal]{IEEEtran}
\usepackage{amsmath,amsfonts}
\usepackage{algorithm}
\usepackage{algpseudocode}
\usepackage{array}
\usepackage{subcaption}
\usepackage{textcomp}
\usepackage{stfloats}
\usepackage{url}
\usepackage{verbatim}
\usepackage{graphicx}
\usepackage{cite}
\usepackage{soul}
\usepackage{xcolor}
\usepackage{tcolorbox}
\usepackage{afterpage}
\usepackage{makecell}
\usepackage{threeparttable}
\usepackage{bm}
\usepackage{dsfont}
\usepackage{xcolor}
\usepackage{graphicx}
\usepackage{mathrsfs}  

\usepackage{amsthm}
\usepackage{tabularx}
\usepackage{amssymb}
\usepackage{enumitem}

\theoremstyle{plain}
\newtheorem{theorem}{Theorem}

\newtheorem{proposition}{Proposition}

\theoremstyle{definition}
\newtheorem{assumption}{Assumption}

\theoremstyle{remark}

\usepackage{mathrsfs}  

\newtheorem{definition}{Definition}

\newtcolorbox{highlighted}{colback=yellow!50,colframe=yellow!50,arc=0pt,outer arc=0pt}

\begin{document}

\title{Semantic-Aware Update Admission and Adaptive Byzantine Redundancy for Secure Computation Offloading in COIN-Enabled Metaverse Systems}

\author{Ibrahim Aliyu,~\IEEEmembership{Member,~IEEE}, Jaehyung Park, and Jinsul Kim~\IEEEmembership{Member,~IEEE}
\thanks{This work was supported in part by the Institute of
Information \& Communications Technology Planning \& Evaluation
(IITP)--ITRC (Information Technology Research Center) grant
funded by the Korea government (MSIT)
(IITP-2026-RS-2024-00437718, 50\%)
and in part by the IITP grant funded by the Korea government
(MSIT) (RS-2025-00345030, Development of digital twin-based
network failure prevention and operation management automation
technology, 50\%).
\textit{Corresponding author: Jinsul Kim
(email: jsworld@jnu.ac.kr).}}%
\thanks{Ibrahim Aliyu and Jinsul Kim are with the
Intelligent Innovation G5-AICT Research Center,
Department of Intelligent Electronics and Computer Engineering,
Chonnam National University, Gwangju 61186,
Republic of Korea.}
\thanks{Jaehyung Park is with the
Department of Intelligent Electronics and Computer Engineering,
Chonnam National University, Gwangju 61186,
Republic of Korea}
}



\maketitle

\begin{abstract}
Computing in the network (COIN) supports latency-sensitive
Metaverse services, but blockchain commitment of persistent shared-state updates competes with application execution for limited resources.  We consider the joint
problem of deciding which updates to commit, how much
Byzantine redundancy to provision, and where to execute
Metaverse application tasks. We formulate this problem
under resource, latency, and practical Byzantine fault
tolerance (PBFT) constraints, accounting for update
persistence and semantic significance. We develop a hierarchical redundancy-aware PBFT
(R-PBFT) approach that decomposes the decision process
into slot-wise computation offloading and sequential
update-admission and redundancy control. We formulate
the former as an identical-interest game and address
the latter using a feasibility-filtered double deep Q-network (DDQN). Analytical downlink allocation
accounts for the energy and delivery-feasibility
effects of consensus latency.
For fixed slot state and redundancy, we establish the
game's exact-potential property and finite termination
of strict sequential improvements. Feasible-action
filtering enforces configured Byzantine redundancy and
blockchain-budget constraints independently of learning.
Numerical results show relative improvements of up to
$46.7\%$ in reward averaged across tested user loads
and seeds over Minimum Secure, the strongest baseline
on average. At the highest tested Byzantine bound, the secure timely-update ratio improves by $62.3\%$ over the same baseline. Ablations assess the contributions of
feasibility filtering and adaptive deferral.
\end{abstract}

\begin{IEEEkeywords}
Blockchain, computation offloading, deep
reinforcement learning, game theory, in-network computing,
Metaverse.
\end{IEEEkeywords}

\section{Introduction}
\label{sec:introduction}
\IEEEPARstart{E}{xtended} reality (XR) applications integrate latency-sensitive functions—such as tracking, simultaneous localization and mapping (SLAM), object detection, gesture processing, rendering, and multimedia encoding—each with unique computation requirements and dependencies. These tasks are efficiently represented as directed acyclic graphs (DAGs) and distributed across user and edge resources\cite{alriksson2021xr,yu2024dependent}. In multi-user XR, maintaining consistent shared spatial and scene information highlights the value of computing in the network (COIN), which brings computation closer to users through distributed resources\cite{chen2023networked,huang2024scaxr,kunze2025coin}.

However, distributed XR processing introduces integrity challenges: persistent shared-state updates affect multiple XR users
and downstream computations, making their manipulation
a significant integrity
threat~\cite{slocum2024sharedstate}.
Byzantine fault tolerance (BFT) maintains agreement on accepted updates with
$r\ge3\bar z+1$ validators under its protocol
assumptions~\cite{lamport2019byzantine,castro1999pbft}. However, input validity requires application-level
checks, while redundant validation increases resource
and latency overhead.

While recent BFT mechanisms adapt consensus to changing
network and infrastructure
conditions~\cite{ye2024labft,ling2024stable,leduc2025draftee},
how shared-state persistence and semantic significance
should jointly determine update commitment,
security-feasible redundancy provisioning, and XR computation
offloading remains an open question (see Table~\ref{tab:related_comparison}). Addressing this
coupling requires balancing state freshness, consensus
overhead, and application latency under a shared resource
budget. 

We formulate the joint optimization of persistent-update
commitment, blockchain redundancy, and XR computation
offloading under shared-resource and PBFT feasibility
constraints. We develop a hierarchical redundancy-aware
PBFT (R-PBFT) approach that decomposes the decision process
into slot-wise offloading and sequential update-admission
and redundancy control. The former is formulated as an
identical-interest game; the latter is addressed using a
feasibility-filtered double deep Q-network (DDQN).

The main contributions are summarized as follows:
\begin{itemize}
\item
We develop a function-level XR DAG model distinguishing
user-specific, shared-intermediate, and persistent
authoritative state, enabling selective blockchain
protection without validating transient outputs.

\item
We introduce persistence-aware commitment and deferral
with adaptive R-PBFT redundancy under a shared resource
budget. Every admitted update satisfies
$r_p^{(t)}\ge3\bar z_p^{(t)}+1$, using the configured
Byzantine-fault bound.

\item
We develop a hierarchical solution using an identical-interest
offloading game and a DDQN admission--redundancy controller.
Joint feasible-action selection enforces configured Byzantine
and budget constraints during action selection and target
construction. Analytical downlink allocation accounts for
consensus latency in delivery energy and feasibility.
For fixed slot state and BRF, we establish the game's
exact-potential property and finite termination of
strict improvements.

\item
We evaluate the framework across Byzantine severity,
user load, shared-state validity, arrival intensity,
and computing capacity, supported by learning diagnostics,
ablations, and statistical analyses.
\end{itemize}

\begin{table}[t]
\centering
\caption{Comparison of Related Research Directions}
\label{tab:related_comparison}
\renewcommand{\arraystretch}{1.08}
\setlength{\tabcolsep}{3pt}
\footnotesize

\resizebox{\columnwidth}{!}{%
\begin{tabular}{|l|c|c|c|c|c|}
\hline
\textbf{Research Direction}
&
\makecell{\rotatebox{90}{XR DAG}}
&
\makecell{\rotatebox{90}{Shared-State Semantics}}
&
\makecell{\rotatebox{90}{Adaptive BFT}}
&
\makecell{\rotatebox{90}{Hard BFT Feasibility}}
&
\makecell{\rotatebox{90}{BC--Compute Coupling}}
\\
\hline

XR DAG offloading
\cite{yu2024dependent}
&
$\checkmark$
&
$\times$
&
$\times$
&
$\times$
&
$\times$
\\
\hline

Multi-user XR and shared-state
\cite{chen2023networked,huang2024scaxr,slocum2024sharedstate}
&
$\triangle$
&
$\checkmark$
&
$\times$
&
$\times$
&
$\times$
\\
\hline

Adaptive BFT
\cite{ye2024labft,ling2024stable,leduc2025draftee}
&
$\times$
&
$\times$
&
$\checkmark$
&
$\triangle$
&
$\times$
\\
\hline

Blockchain-assisted edge offloading
\cite{fan2024resource,fang2024task,du2025secure,
chen2025blockchain}
&
$\triangle$
&
$\times$
&
$\triangle$
&
$\triangle$
&
$\checkmark$
\\
\hline

\textbf{Proposed R-PBFT}
&
$\checkmark$
&
$\checkmark$
&
$\checkmark$
&
$\checkmark$
&
$\checkmark$
\\
\hline

\end{tabular}%
}

\vspace{1mm}
\parbox{\columnwidth}{%
\scriptsize
\textit{Note:}
$\checkmark$: explicitly addressed;
$\triangle$: partially addressed or present in some
cited works; $\times$: not explicitly addressed in
the cited works.
}

\end{table}

The paper is organized as follows. Sections~\ref{sec:system_model}--\ref{sec:security_problem}
present the system model and security-constrained formulation.
Section~\ref{sec:proposed_solution} develops the proposed solution,
and Section~\ref{sec:numerical} evaluates its performance.
Section~\ref{sec:related_work} discusses related work,
and Section~\ref{sec:conclusion} concludes the paper.

\begin{figure*}[!t]
    \centering
    \includegraphics[
    width=0.9\textwidth,
    height=0.5 \textheight,
    keepaspectratio]
    {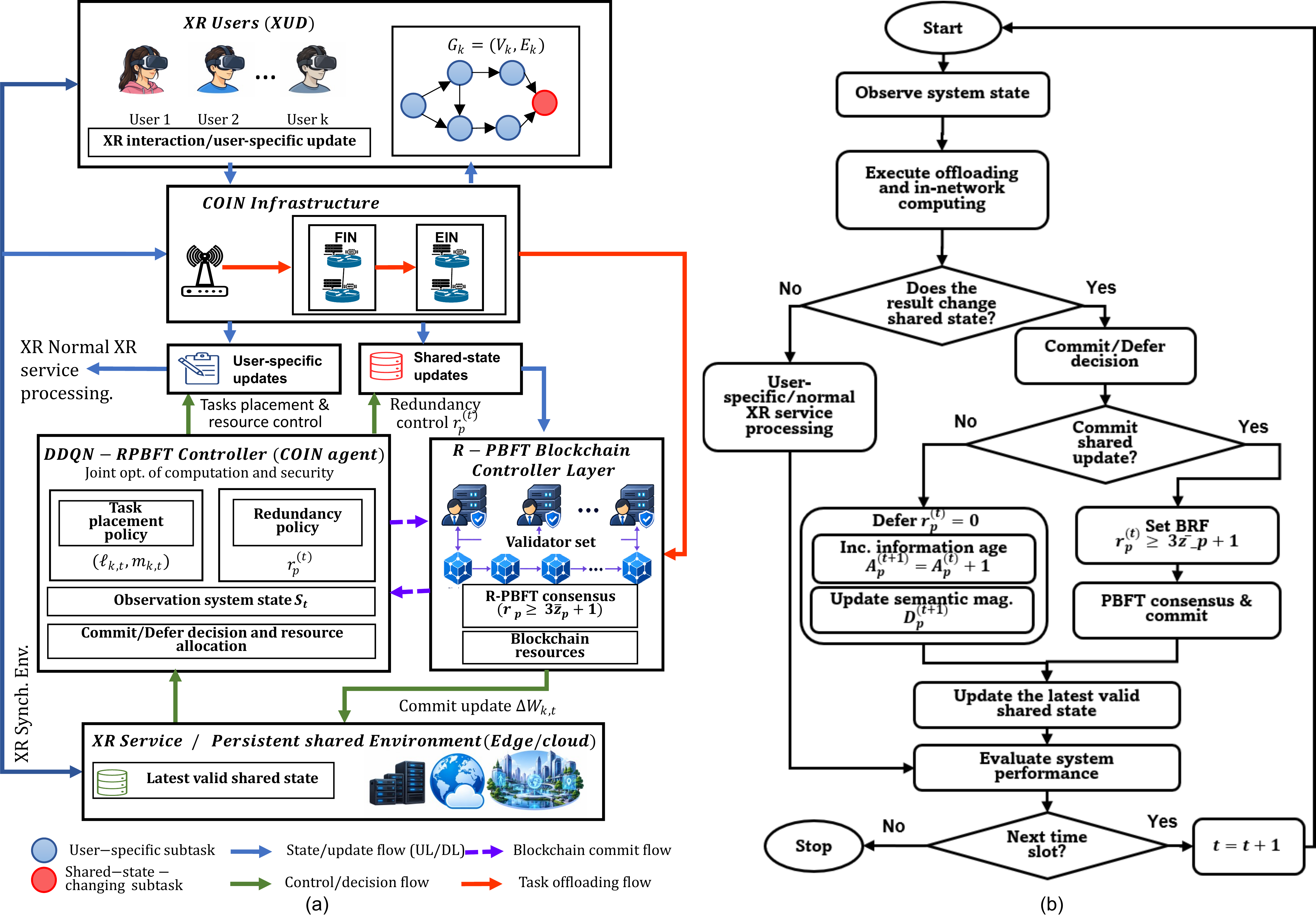}
   \caption{DDQN-RPBFT framework:
(a) COIN-enabled XR architecture with persistent shared-state
protection; (b) Per-slot operation comprising offloading,
commit/defer decisions, adaptive BRF selection, and
R-PBFT commitment.}
    \label{fig:system_architecture}
\end{figure*}

\section{System Model}
\label{sec:system_model}

\subsection{Network Architecture}
\label{subsec:network_architecture}
We consider a COIN-enabled Metaverse system comprising
users $\mathcal K=\{1,\ldots,K\}$, fog in-network
computing nodes (FINs) $\mathcal N^{\mathrm F}$,
edge in-network computing nodes (EINs)
$\mathcal N^{\mathrm E}$, permissioned validators
$\mathcal V$, and a logically centralized controller.
FINs and EINs comprise programmable network devices
and associated computing hosts~\cite{kunze2025coin}.
Execution resources support their assigned functions;
validators require persistent storage.

Time is divided into slots of duration $\Delta t$.
Each slot, the controller selects a blockchain redundancy
factor (BRF) profile using task demand, shared-state age
and magnitude, configured fault bounds, and available
resources. The profile remains fixed during offloading
and channel selection.

\subsection{Task Model}
\label{subsec:task_model}

User $k$ requests at most one XR task per slot:
$\mu_k^{(t)}\in\{0\}\cup\mathcal F$, where $0$ denotes
inactivity and $\mathcal F=\{1,\ldots,F\}$ contains task
types. Active users form
$\mathcal K_t=\{k:\mu_k^{(t)}\neq0\}$.
Task $f$ is a DAG $\mathcal G_f=(\mathcal J_f,\mathcal E_f)$
of functions and dependencies.

Each function has attributes
$\mathcal Q_{f,j}=\langle I_{f,j},C_{f,j},O_{f,j},
a_{f,j}^{\mathrm{sh}},a_{f,j}^{\mathrm{per}}\rangle$,
where $I_{f,j}$ and $O_{f,j}$ are input/output sizes
in bytes, $C_{f,j}$ is computation demand in cycles,
and the binary indicators identify shared processing
and persistent-update generation.
Only persistent authoritative updates are eligible for
blockchain commitment.

Relative to DAG execution start, function times satisfy
\begin{align}
T_{k,j}^{\mathrm{st}}(t)
&=\max_{i\in\operatorname{pred}(j)}
\left\{T_{k,i}^{\mathrm{fn}}(t)
+T_{k,i\rightarrow j}^{\mathrm{tr}}(t)\right\},
\label{eq:function_start}\\
T_{k,j}^{\mathrm{fn}}(t)
&=T_{k,j}^{\mathrm{st}}(t)+T_{k,j}^{\mathrm{exe}}(t),
\label{eq:function_finish}
\end{align}
where $\operatorname{pred}(j)$ denotes the immediate
predecessors of $j$; $T^{\mathrm{st}}$, $T^{\mathrm{fn}}$,
$T^{\mathrm{tr}}$, and $T^{\mathrm{exe}}$ denote start
time, finish time, inter-function transfer delay, and
execution duration, respectively. Source functions
start at zero; resource waiting is accounted for separately. The DAG latency is
\begin{equation}
T_{k,f}^{\mathrm{DAG}}(t)
=\max_{j\in\mathcal J_f}T_{k,j}^{\mathrm{fn}}(t).
\label{eq:dag_completion}
\end{equation}

\subsection{Blockchain Model}
\label{subsec:blockchain_model}

Let $\mathcal P$ contain persistent shared functions. The pending
update associated with function $p$ is characterized by $\mathcal U_p^{(t)}
=
\left\langle
W_p,A_p^{(t)},D_p^{(t)},\chi_p,\bar z_p^{(t)}
\right\rangle,$ where $W_p$ is its size, $A_p^{(t)}$ is the shared-state age
proxy, $D_p^{(t)}\in[0,1]$ is the pending semantic magnitude,
$\chi_p>0$ is criticality, and $\bar z_p^{(t)}$ is the configured
Byzantine-fault bound.

The affected-user set and active-update set are
$\mathcal K_p^{(t)}
=\{k\in\mathcal K_t:p\in\mathcal J_{\mu_k^{(t)}}\}$ and
$\mathcal P_t^{\mathrm{act}}
=\{p:|\mathcal K_p^{(t)}|>0\}$, respectively, where $N_p^{(t)}=|\mathcal K_p^{(t)}|$
is the number of affected users.

The controller selects the BRF profile
$\mathbf r_t=(r_p^{(t)})_{p\in\mathcal P}$, with
\begin{equation}
r_p^{(t)}\in\mathcal R_p^{(t)}
=\{0\}\cup
\{r\in\mathbb Z:
3\bar z_p^{(t)}+1\le r\le r_{\max}\},
\label{eq:brf_action_set}
\end{equation}
where $r_{\max}\le|\mathcal V|$ and
$r_{p,\min}^{(t)}=3\bar z_p^{(t)}+1$.
Zero denotes deferral; a positive action admits a
commitment attempt using $r_p^{(t)}$ validators.
If the positive-action set is empty, only deferral
is available.

The slot-level transition model assumes successful commitment
for every admitted security-feasible update; timely delivery
to users is evaluated separately. With
$c_p^{(t)}=\mathds{1}\{r_p^{(t)}>0\}$,
\begin{align}
A_p^{(t+1)}
&=(1-c_p^{(t)})(A_p^{(t)}+\Delta t),
\label{eq:aoi_evolution}\\
D_p^{(t+1)}
&=
\begin{cases}
\varrho_pD_p^{(t)}, & c_p^{(t)}=1,\\
\min\{1,D_p^{(t)}+\xi_p^{(t)}\}, & c_p^{(t)}=0,
\end{cases}
\label{eq:magnitude_evolution}
\end{align}
where $\varrho_p\in[0,1)$ and $\xi_p^{(t)}\ge0$ describe
residual magnitude and accumulated changes. The zero reset abstracts generation-to-commit delay;
$A_p^{(t)}$ is an age proxy rather than continuous-time
receiver Age of Information~\cite{yates2021age}. This transition
assumption is distinct from the consensus safety guarantee.

\subsection{Communication Model}
\label{subsec:communication_model}

\subsubsection{Uplink Transmission}
User $k$ selects $s_k^{(t)}\in\{0\}\cup\mathcal M$,
where $\mathcal M=\{1,\ldots,M\}$; zero denotes local
execution and $m>0$ denotes offloading through channel $m$.
Let $\mathbf s_t=(s_k^{(t)})_{k\in\mathcal K_t}$,
$s_{k,m}^{(t)}=\mathds{1}\{s_k^{(t)}=m\}$, and
$\boldsymbol{\mathcal S}_t$ collect these indicators. The uplink rate is
\begin{equation}
R_{k,m}^{\mathrm{UL}}(t)
=
\frac{B^{\mathrm{UL}}}{M}
\log_2
\left(
1+
\frac{P_k^{\mathrm{UL}}h_k^{\mathrm{UL}}(t)}
{\sigma_{\mathrm{UL}}^2+
\sum_{\substack{n\in\mathcal K_t\setminus\{k\}\\s_n^{(t)}=m}}
P_n^{\mathrm{UL}}h_n^{\mathrm{UL}}(t)}
\right),
\label{eq:uplink_rate}
\end{equation}
where $B^{\mathrm{UL}}$, $P_k^{\mathrm{UL}}$,
$h_k^{\mathrm{UL}}$, and $\sigma_{\mathrm{UL}}^2$ denote
bandwidth, transmit power, channel gain, and noise power.
For remote execution of $f=\mu_k^{(t)}$, the uplink
latency is
$T_{k,f}^{\mathrm{UL}}(t)
=I_f/R_{k,s_k^{(t)}}^{\mathrm{UL}}(t)$,
and the transmit energy is
$E_{k,f}^{\mathrm{UL}}(t)
=P_k^{\mathrm{UL}}T_{k,f}^{\mathrm{UL}}(t)$.
Both are zero for local execution.

\subsubsection{Blockchain Communication}

For a committee of size $r$, let
$\mathcal H=\{\mathrm{br},\mathrm{pp},\mathrm{pr},
\mathrm{cm},\mathrm{rp}\}$ denote dissemination,
pre-prepare, prepare, commit, and reply stages, respectively.
Let $W_p^h$ and $D_p^h(r)$ denote the message size
and aggregate communication volume of stage $h$.
The modeled volumes are
\begin{equation}
\begin{aligned}
D_p^{\mathrm{pp}}(r)&=(r-1)W_p^{\mathrm{pp}},
& D_p^{\mathrm{rp}}(r)&=rW_p^{\mathrm{rp}},\\
D_p^h(r)&=r(r-1)W_p^h,
& h&\in\{\mathrm{pr},\mathrm{cm}\}.
\end{aligned}
\label{eq:pbft_message_volumes}
\end{equation}
Dissemination volume $D_p^{\mathrm{br}}(r)$ excludes
traffic counted in pre-prepare. Prepare traffic uses
an all-validator approximation, not a certificate threshold.

With effective aggregate rate $R^{\mathrm{BC}}$,
processing capacity $F^{\mathrm{BC}}$, and stage
workload $C_p^h(r)$, the nominal consensus latency is
\begin{equation}
T_p^{\mathrm{BC}}(r)
=\sum_{h\in\mathcal H}
\left[
\frac{D_p^h(r)}{R^{\mathrm{BC}}}
+\frac{C_p^h(r)}{F^{\mathrm{BC}}}
\right].
\label{eq:total_pbft_latency}
\end{equation}
Total communication and replicated storage are
$D_p^{\mathrm{BC}}(r)=\sum_{h\in\mathcal H}D_p^h(r)$
and $S_p^{\mathrm{BC}}(r)=rW_p$, respectively.
The model accounts for nominal consensus operation,
excluding view-change and committee-reconfiguration overhead.

\subsubsection{Downlink Transmission}

For bandwidth $B_k^{\mathrm{DL}}$ assigned to user $k$,
\begin{equation}
R_k^{\mathrm{DL}}(t)
=
B_k^{\mathrm{DL}}
\log_2
\left(
1+
\frac{P_k^{\mathrm{DL}}(t)h_k^{\mathrm{DL}}(t)}
{I_k^{\mathrm{DL}}(t)+\sigma_{\mathrm{DL}}^2}
\right),
\label{eq:downlink_rate}
\end{equation}
where $I_k^{\mathrm{DL}}$ is background interference.
The delivery latency is
$T_{k,f}^{\mathrm{DL}}(t)
=O_{k,f}^{\mathrm{DL}}/R_k^{\mathrm{DL}}(t)$,
with transmitter-side energy
$E_{k,f}^{\mathrm{DL}}(t)
=P_k^{\mathrm{DL}}(t)T_{k,f}^{\mathrm{DL}}(t)$.
Downlink power is recomputed for each candidate offloading
profile using the remaining service time, as specified in
Section~\ref{subsubsec:downlink_power_allocation}.

\subsection{Computation and Service Model}
\label{subsec:computation_model}

\subsubsection{Local Computing}

For local processor frequency $F_k^{\mathrm L}(t)$,
the execution latency and energy are:
\begin{align}
T_{k,f}^{\mathrm L}(t)
&=T_{k,f}^{\mathrm{DAG}}\bigl(F_k^{\mathrm L}(t)\bigr),
\label{eq:local_latency}\\
E_{k,f}^{\mathrm L}(t)
&=\kappa_kC_f\bigl(F_k^{\mathrm L}(t)\bigr)^2,
\label{eq:local_energy}
\end{align}
where $\kappa_k$ is the switched-capacitance coefficient and
$C_f=\sum_{j\in\mathcal J_f}C_{f,j}$.

\subsubsection{Remote Computing}

For execution at $\ell\in\{\mathrm{FIN},\mathrm{EIN}\}$,
as determined by the placement rule, the latency
excluding blockchain validation is
\begin{equation}
T_{k,f}^{\ell}(t)
=
T_{k,f}^{\mathrm{UL}}(t)
+T_{k,f}^{q,\ell}(t)
+T_{k,f}^{\mathrm{DAG},\ell}(t)
+T_{k,f}^{\mathrm{DL}}(t).
\label{eq:remote_execution_latency}
\end{equation}

\subsubsection{Service latency and user energy}
Let $\nu_{k,p}^{(t)}=\mathds{1}\{k\in\mathcal K_p^{(t)}\}$.
Under the modeled serialized validation waits,
\begin{equation}
T_k^{\mathrm{val}}(t)
=
\sum_{p\in\mathcal P}
\nu_{k,p}^{(t)}c_p^{(t)}
T_p^{\mathrm{BC}}\bigl(r_p^{(t)}\bigr).
\label{eq:user_validation_wait}
\end{equation}

The physical service latency is
\begin{equation}
T_k^{\mathrm{phys}}(t)
=
\begin{cases}
T_{k,f}^{\mathrm L}(t)+T_k^{\mathrm{val}}(t)
+T_{k,f}^{\mathrm{DL}}(t),
& s_k^{(t)}=0,\\
T_{k,f}^{\ell}(t)+T_k^{\mathrm{val}}(t),
& s_k^{(t)}>0.
\end{cases}
\label{eq:e2e_latency}
\end{equation}
Only applicable delivery and validation components
are included; local execution does not eliminate
dependencies on shared state.

User-device energy is
\begin{equation}
E_k^{\mathrm U}(t)
=
E_k^{\mathrm{loc}}(t)
+E_k^{\mathrm{UL}}(t)
+E_k^{\mathrm{rx}}(t),
\label{eq:user_device_energy}
\end{equation}
 where $E_k^{\mathrm{rx}}(t)$ denotes user-device energy
for receiving downlink data; inactive components are zero.
Remote computation and downlink transmission energy
enter the operator cost

\section{Security-Constrained Problem Formulation}
\label{sec:security_problem}

\subsection{Security Analysis}
\label{subsec:security_analysis}

\subsubsection{Threat Model and Validation}
We assume authenticated, partially synchronous consensus
and secure cryptography. Compromised users and computing
nodes may submit false inputs; Byzantine validators may
equivocate or disrupt consensus.

\begin{assumption}
\label{ass:bounded_byzantine}
Each configuration contains at most $\bar z_p^{(t)}$
Byzantine validators; membership and this bound remain
fixed through view changes.
\end{assumption}

Correct validators accept update $u$ in state $S$ only if
\begin{equation}
\begin{aligned}
\mathsf{Valid}(S,u)
={}&\mathsf{Auth}(u)\land\mathsf{Permit}(S,u)\\
&\land\mathsf{Version}(S,u)\land\mathsf{Rule}(S,u),
\end{aligned}
\label{eq:application_validity}
\end{equation}
checking signatures, permissions, version/replay consistency,
and deterministic application rules. State transitions are
deterministic; false inputs passing validation remain
outside the guarantee.

\subsubsection{Consensus Configuration}
For $n=r_p^{(t)}$ and $b=\bar z_p^{(t)}$, the
certificate threshold is
\begin{equation}
q=n-b,\qquad n\ge3b+1.
\label{eq:adaptive_quorum_threshold}
\end{equation}
Messages bind stream, configuration, view, sequence,
and value. Correct validators retain protocol state
and do not equivocate. Preparation requires a primary
pre-prepare and $q-1$ matching distinct backup prepares.
Only prepared validators send commits; commitment
requires preparation and $q$ distinct matching commits.

New views collect $q$ reports and verifiably preserve
the highest-view prepared value per sequence above a
$q$-certified checkpoint, using nulls only where none
is reported. Reconfiguration requires an outgoing commit
certificate binding the terminal state and successor
configuration, plus $q$ state attestations. Outgoing
validators stop at the boundary; successors verify the
certificate chain and transferred state before extending
the history. Historical certificates remain unforgeable.

\begin{theorem}
\label{thm:adaptive_brf_safety}
Under Assumption~\ref{ass:bounded_byzantine} and these
conditions, adaptive BRF preserves agreement within
and across certified configurations.
\end{theorem}

\begin{proof}
For signer sets $|Q_1|,|Q_2|\ge q$,
\begin{equation}
|Q_1\cap Q_2|\ge2q-n=n-2b\ge b+1,
\label{eq:adaptive_quorum_intersection}
\end{equation}
excluding conflicting preparations within a view.
Every set of $q$ view-change reporters intersects the
$q-b$ correct prepared committers of a committed value,
since $2q-b-n\ge1$. Verified selection preserves agreement
inductively across views; certified handoffs preserve
the committed history across configurations.
\end{proof}

Agreement guarantees neither input truth nor timely,
fresh delivery. Deferral starts no new attempt and
does not cancel ongoing consensus.

\subsection{Cost Model}
\label{subsec:cost_model}

Costs depend on the current offloading and BRF profiles;
all tradeoff weights are nonnegative and fixed.
\subsubsection{User-Oriented Cost}
Let $\mathcal P_{k,t}$ contain user $k$'s required
persistent updates. The indicator $F_{p,k,t}$ equals one when an attempted
delivery satisfies the configured Byzantine requirement,
downlink feasibility, and generation-to-delivery freshness
deadline, and zero otherwise. Define
\begin{align}
A_{k,t}
&=\frac{\sum_{p\in\mathcal P_{k,t}}A_p^{(t)}}
{\max\{|\mathcal P_{k,t}|,1\}},
\label{eq:user_mean_state_age}\\
Z_{k,t}
&=\mathds{1}\{\mathcal P_{k,t}\neq\varnothing\}
\prod_{p\in\mathcal P_{k,t}}F_{p,k,t},
\label{eq:complete_update_delivery}\\
T_k^{\mathrm{eff}}(t)
&=T_k^{\mathrm{phys}}(t)+(1-Z_{k,t})A_{k,t},
\label{eq:effective_service_latency}\\
v_k^{(t)}
&=\mathds{1}\{T_k^{\mathrm{eff}}(t)>T_k^{\max}\}.
\label{eq:sla_violation_indicator}
\end{align}
Here, $A_{k,t}$ is the mean required-state age,
$Z_{k,t}$ indicates successful delivery of all required
updates for a nonempty set, $T_k^{\mathrm{eff}}$ is
age-penalized service latency, and $v_k^{(t)}$ indicates
an effective-latency deadline violation. The added age
is a penalty, not physical delay; for an empty update
set, $A_{k,t}=Z_{k,t}=0$. User cost is
\begin{equation}
C_k^{\mathrm U}(t)
=w_TT_k^{\mathrm{eff}}(t)
+w_EE_k^{\mathrm U}(t)+w_AA_{k,t}.
\label{eq:current_user_cost}
\end{equation}

\noindent where $w_T$, $w_E$, and $w_A$ weight effective service
latency, user energy, and shared-state age, respectively.

\subsubsection{Commitment Cost}
For positive feasible $r$, blockchain energy and cost are
\begin{align}
E_p^{\mathrm{BC}}(r)
&=e_{\mathrm{com}}D_p^{\mathrm{BC}}(r)
+e_{\mathrm{cpu}}\sum_{h\in\mathcal H}C_p^h(r)
+e_{\mathrm{sto}}S_p^{\mathrm{BC}}(r),
\label{eq:blockchain_energy}\\
C_p^{\mathrm{BC}}(r)
&=w_{\mathrm{BC}}T_p^{\mathrm{BC}}(r)+E_p^{\mathrm{BC}}(r),
\label{eq:update_blockchain_cost}
\end{align}
where $w_{\mathrm{BC}}$ weights validation latency;
$e_{\mathrm{com}}$, $e_{\mathrm{cpu}}$, and
$e_{\mathrm{sto}}$ specify energy per communicated bit,
processing cycle, and newly stored bit, respectively.
The redundancy-margin penalty
\begin{equation}
\ell_{p,t}(r)=\frac{\chi_p}{1+r-r_{p,\min}^{(t)}}
\label{eq:update_headroom_loss}
\end{equation}
favors additional redundancy for critical updates with
diminishing reductions; it is not a failure probability.

With shared-delivery energy $E_{p,t}^{\mathrm{DL,sh}}$
and unsuccessful fresh-delivery count $N_{p,t}^{\mathrm{miss}}$,
the admitted-update cost is
\begin{equation}
\begin{aligned}
G_{p,t}^{+}(r)
={}&C_p^{\mathrm{BC}}(r)+w_{\mathrm L}\ell_{p,t}(r)\\
&+w_{\mathrm{DL}}E_{p,t}^{\mathrm{DL,sh}}
+w_{\mathrm V}N_{p,t}^{\mathrm{miss}}.
\end{aligned}
\label{eq:admitted_update_cost}
\end{equation}

\noindent where $w_{\mathrm L}$, $w_{\mathrm{DL}}$, and
$w_{\mathrm V}$ weight the redundancy-margin penalty,
downlink energy, and unsuccessful fresh deliveries,
respectively.

\subsubsection{Update-Deferral Cost}
For validity limit $A_p^{\max}=H_p\Delta t>0$, define
$a_{p,t}=A_p^{(t)}/A_p^{\max}$,
$n_{p,t}=N_p^{(t)}/K$, and
$\widehat\chi_p=\max\{\chi_p,0.1\}$.
Motivated by nonlinear age penalties~\cite{sun2019sampling}
and semantics-aware freshness~\cite{maatouk2023incorrect},
the deferral loss is
\begin{equation}
\begin{aligned}
L_p^{\mathrm{def}}(t)
={}&w_{\mathrm{def}}N_p^{(t)}
\widehat\chi_pD_p^{(t)}\\
&\times\left(1-e^{-\lambda_p a_{p,t}D_p^{(t)}
n_{p,t}\widehat\chi_p}\right).
\end{aligned}
\label{eq:current_deferral_loss}
\end{equation}
 where $w_{\mathrm{def}}\ge0$ scales the loss, and
$\lambda_p>0$ controls the sensitivity of the increasing,
saturating exposure score to age, magnitude, sharing,
and criticality. The exposure score is a modeling assumption, not
a calibrated failure probability.
The total deferral cost is

\begin{equation}
\begin{aligned}
G_{p,t}^{0}
={}&w_{a,1}a_{p,t}+w_{a,2}a_{p,t}^{2}
+w_dD_p^{(t)}+w_nn_{p,t}\\
&+w_i a_{p,t}D_p^{(t)}n_{p,t}\widehat\chi_p
+L_p^{\mathrm{def}}(t)\\
&+w_o\mathds{1}\{a_{p,t}\ge1\}.
\end{aligned}
\label{eq:complete_deferral_cost}
\end{equation}

\noindent where $w_{a,1}$ and $w_{a,2}$ weight linear and quadratic
age, while $w_d$, $w_n$, $w_i$, and $w_o$ weight
magnitude, affected-user fraction, the joint
age--magnitude--sharing--criticality term, and
obsolescence, respectively.

\subsubsection{System Operating Cost}
The shared-update cost is
\begin{equation}
C_{p,t}^{\mathrm{shared}}
=
\begin{cases}
G_{p,t}^{0}, & r_p^{(t)}=0,\\
G_{p,t}^{+}(r_p^{(t)}), & r_p^{(t)}>0.
\end{cases}
\label{eq:shared_update_cost}
\end{equation}
For operator computation and downlink energies
$E_t^{\mathrm C}$ and $E_t^{\mathrm{DL}}$,
the operator cost is
\begin{equation}
\begin{aligned}
C_t^{\mathrm O}
={}&w_{\mathrm C}E_t^{\mathrm C}
+w_{\mathrm{DL}}E_t^{\mathrm{DL}}
+w_{\mathrm V}\sum_{k\in\mathcal K_t}v_k^{(t)}\\
&+\sum_{p\in\mathcal P_t^{\mathrm{act}}}
C_{p,t}^{\mathrm{shared}}.
\end{aligned}
\label{eq:operator_operating_cost}
\end{equation}
\noindent where $w_{\mathrm C}$ weights computation energy,
while $w_{\mathrm{DL}}$ and $w_{\mathrm V}$ weight
downlink energy and deadline violations, respectively. Aggregating redundancy-margin penalties gives
\begin{equation}
L_t^{\mathrm{risk}}
=\sum_{\substack{p\in\mathcal P_t^{\mathrm{act}}\\r_p^{(t)}>0}}
\ell_{p,t}\bigl(r_p^{(t)}\bigr).
\label{eq:reward_modeled_loss}
\end{equation}
The per-slot system operating cost combines user costs,
operator cost, and the aggregate redundancy-margin penalty:
\begin{equation}
\begin{aligned}
C^{\mathrm{sys}}(t)
&=J_t(\mathbf s_t;\mathbf r_t)\\
&=\beta_{\mathrm U}\sum_{k\in\mathcal K_t}C_k^{\mathrm U}(t)
+\beta_{\mathrm O}C_t^{\mathrm O}
+\beta_{\mathrm R}L_t^{\mathrm{risk}}.
\end{aligned}
\label{eq:system_operating_cost}
\end{equation}

The weights $\beta_{\mathrm U}$, $\beta_{\mathrm O}$,
and $\beta_{\mathrm R}$ balance user cost, operator
cost, and redundancy-margin penalties, respectively.

\subsection{Problem Formulation}
\label{subsec:problem_formulation}

Offloading and blockchain commitment compete for limited resources, while update deferral affects subsequent shared-state age and service cost. Thus, we jointly optimize
offloading and BRF decisions over causal policies to
minimize the long-run average system cost:

\begin{align}
(P1):\quad
\min_{\{\boldsymbol{\mathcal S}_t,\mathbf r_t\}_{t\ge1}}
\quad &
\limsup_{T\rightarrow\infty}
\frac{1}{T}\sum_{t=1}^{T}
\mathbb E\!\left[C^{\mathrm{sys}}(t)\right]
\label{eq:joint_problem}\\
\mathrm{s.t.}\quad
&s_{k,m}^{(t)}\in\{0,1\},
\quad \forall k,m,t,
\tag{C1}\label{con:binary_offloading}\\
&\sum_{m\in\mathcal M\cup\{0\}}s_{k,m}^{(t)}=1,
\quad \forall k,t,
\tag{C2}\label{con:single_execution}\\
&s_{k,0}^{(t)}=1,
\quad \forall k\notin\mathcal K_t,\ \forall t,
\tag{C3}\label{con:idle_user}\\
&r_p^{(t)}\in\mathcal R_p^{(t)},
\quad \forall p,t,
\tag{C4}\label{con:brf_domain}\\
&\begin{aligned}
&\mathcal V_p^{(t)}\subseteq\mathcal V,\qquad
|\mathcal V_p^{(t)}|=r_p^{(t)},\\
&\qquad \forall t,\ \forall p:r_p^{(t)}>0,
\end{aligned}
\tag{C5}\label{con:committee_size}\\
&\sum_{p\in\mathcal P}w_{p,r_p^{(t)}}\le B_{\mathrm{BC}},
\quad \forall t,
\tag{C6}\label{con:blockchain_budget}\\
&\begin{aligned}
&T_{k,j}^{\mathrm{st}}(t)
\ge T_{k,i}^{\mathrm{fn}}(t)
+T_{k,i\rightarrow j}^{\mathrm{tr}}(t),\\
&\qquad
\forall(i,j)\in\mathcal E_{\mu_k^{(t)}},
\ \forall k\in\mathcal K_t,\ \forall t,
\end{aligned}
\tag{C7}\label{con:dag_precedence}\\
&0\le P_k^{\mathrm{DL}}(t)\le P_{\max}^{\mathrm{DL}},
\quad \forall k,t.
\tag{C8}\label{con:downlink_power}
\end{align}

Here, $w_{p,0}=0$ and
$w_{p,r}=D_p^{\mathrm{BC}}(r)+S_p^{\mathrm{BC}}(r)$
for $r>0$, with both volumes and $B_{\mathrm{BC}}$
expressed in bits. The budget accounts for communication
and newly replicated storage per slot.
Constraints~\eqref{con:binary_offloading}--\eqref{con:idle_user}
specify execution decisions;
\eqref{con:brf_domain}--\eqref{con:blockchain_budget}
enforce security-feasible admission, committee size,
and the blockchain budget.
Constraint~\eqref{con:dag_precedence} enforces task
dependencies, while \eqref{con:downlink_power} bounds
downlink transmit power.

Computing capacity affects execution and queueing
delays; the infrastructure experiments additionally
penalize overload rather than impose a hard per-slot
workload limit.
State transitions follow
\eqref{eq:aoi_evolution}--\eqref{eq:magnitude_evolution},
while service deadlines are penalized through
\eqref{eq:sla_violation_indicator}.

The following proposition characterizes the static
additive admission subproblem underlying
budget-constrained BRF selection.

\begin{proposition}
\label{prop:nphard}
The static additive update-admission problem associated
with $(P1)$ is NP-hard when positive integer resource
weights and nonnegative integer commitment savings
are independently specified inputs.
\end{proposition}

\begin{proof}
Consider a $0$--$1$ knapsack instance with positive
integer weights $w_p$, values $v_p$, and capacity
$B_{\mathrm{BC}}$. Fix offloading and restrict each
update to deferral or minimum-secure commitment.
Let $x_p\in\{0,1\}$ indicate commitment, consuming
$w_px_p$ resources. Assign commitment cost $c_p=1$
and deferral cost $c_p+v_p$. Then
\begin{equation}
\begin{aligned}
C(\mathbf x)
&=\sum_p\bigl[x_pc_p+(1-x_p)(c_p+v_p)\bigr]\\
&=\sum_p(c_p+v_p)-\sum_pv_px_p.
\end{aligned}
\label{eq:admission_cost_decomposition}
\end{equation}
Thus, minimizing $C(\mathbf x)$ subject to
$\sum_pw_px_p\le B_{\mathrm{BC}}$ is equivalent to
maximizing $\sum_pv_px_p$ under the same constraint.
This polynomial-time reduction establishes NP-hardness
of the stated static additive admission problem.
\end{proof}

This knapsack structure motivates the budget-constrained
BRF selection developed next.

\section{Proposed Solution}
\label{sec:proposed_solution}

Problem $(P1)$ couples discrete offloading and redundancy
decisions through interference, consensus overhead, and
shared resources. Commitment consumes application-service
resources, whereas deferral affects subsequent freshness.
We decompose the decision process into sequential admission
and redundancy control and conditional slot-wise offloading.
The upper layer selects $\mathbf r_t$ using DDQN; the lower
layer computes $\widehat{\mathbf s}_t(\mathbf r_t)$ through
centralized sequential updates in an identical-interest game:
$\mathbf r_t
\longrightarrow\widehat{\mathbf s}_t(\mathbf r_t)
\longrightarrow(R_t,\mathbf x_{t+1})
\longrightarrow\mathbf r_{t+1}$.
The lower layer reduces the common operating cost,
while the upper layer learns admission and redundancy
decisions from a discounted reward balancing service
performance, resource expenditure, and update freshness.

\subsection{Multiuser Offloading Game}
\label{subsec:multiuser_offloading_game}

A user's offloading decision affects other users through
interference and shared resources. To account for these
effects, the centralized controller evaluates unilateral
changes using the system operating cost $J_t$ in
\eqref{eq:system_operating_cost}, defining an
identical-interest game
\cite{monderer1996potential,yamamoto2015comprehensive}.

During lower-layer optimization, the BRF profile and
exogenous slot state remain fixed, while interference,
downlink power, and costs are recomputed for each
candidate profile. Using the previously defined actions,
let $\mathcal S_{k,t}=\{0\}\cup\mathcal M$,
$\mathcal S_t=\prod_{k\in\mathcal K_t}\mathcal S_{k,t}$,
and $s_{k,t}\equiv s_k^{(t)}$.

\begin{definition}
\label{def:multiuser_offloading_game}
For fixed $\mathbf r_t$, the offloading game is
$\mathcal G_t(\mathbf r_t)
=\langle\mathcal K_t,\{\mathcal S_{k,t}\},
\{\widetilde C_{k,t}\}\rangle$, with common decision cost
\begin{equation}
\widetilde C_{k,t}(\mathbf s_t;\mathbf r_t)
=J_t(\mathbf s_t;\mathbf r_t),
\qquad k\in\mathcal K_t.
\label{eq:user_game_cost}
\end{equation}
\end{definition}

The individual user cost $C_k^{\mathrm U}$ contributes
to $J_t$, whereas $\widetilde C_{k,t}$ includes the
system-wide consequences of a strategy change.
The conditional offloading problem is
\begin{equation}
(P2):\quad
\min_{\mathbf s_t\in\mathcal S_t}
J_t(\mathbf s_t;\mathbf r_t).
\label{eq:slot_offloading_problem}
\end{equation}
Sequential best responses seek a coordinatewise minimum,
without guaranteeing global optimality.

The candidate uplink rate follows \eqref{eq:uplink_rate}.
Let $C_{k,t}^{\mathrm L}(0,\mathbf s_{-k,t};\mathbf r_t)$
denote user $k$'s cost under local execution.
For offloading through channel $m$, the corresponding
user cost is
\begin{equation}
\begin{aligned}
&C_{k,t}^{\mathrm O}(m,\mathbf s_{-k,t};\mathbf r_t)\\
&\quad=
\frac{\alpha_{k,t}}
{R_{k,m}^{\mathrm{UL}}(m,\mathbf s_{-k,t})}
+G_{k,t}^{\mathrm O}(m,\mathbf s_{-k,t};\mathbf r_t),
\end{aligned}
\label{eq:remote_decision_cost}
\end{equation}
where $\alpha_{k,t}=I_{\mu_k^{(t)}}(w_T+w_EP_k^{\mathrm{UL}})$,
and $G_{k,t}^{\mathrm O}$ contains the remaining
user-cost contributions.

For brevity, let
$C_{k,t}^{\mathrm L}
=C_{k,t}^{\mathrm L}(0,\mathbf s_{-k,t};\mathbf r_t)$,
$G_{k,t}^{\mathrm O}(m)
=G_{k,t}^{\mathrm O}(m,\mathbf s_{-k,t};\mathbf r_t)$,
and $R_{k,m}=R_{k,m}^{\mathrm{UL}}(m,\mathbf s_{-k,t})$.

\begin{proposition}
\label{thm:current_ordinal_potential}
For fixed exogenous slot state and BRF profile, suppose
$J_t$ is finite and deterministic on the finite, nonempty
strategy space $\mathcal S_t$. The identical-interest
game admits the exact potential
\begin{equation}
\begin{aligned}
\phi_t(\mathbf s_t;\mathbf r_t)
={}&\beta_{\mathrm U}\sum_{k\in\mathcal K_t}
\Bigg[
s_{k,0}^{(t)}C_{k,t}^{\mathrm L}\\
&+\sum_{m\in\mathcal M}s_{k,m}^{(t)}
\left(
\frac{\alpha_{k,t}}{R_{k,m}}
+G_{k,t}^{\mathrm O}(m)
\right)
\Bigg]\\
&+\beta_{\mathrm O}C_t^{\mathrm O}
+\beta_{\mathrm R}L_t^{\mathrm{risk}}\\
={}&J_t(\mathbf s_t;\mathbf r_t),
\end{aligned}
\label{eq:current_ordinal_potential}
\end{equation}
and has a pure Nash equilibrium.
\end{proposition}

\begin{proof}
Fix $t$, $\mathbf r_t$, and $\mathbf s_{-k,t}$,
and suppress these fixed arguments. For a candidate
action $a\in\mathcal S_{k,t}$, define
\begin{equation}
\begin{aligned}
\phi_{-k,t}(a)
={}&\beta_{\mathrm U}
\sum_{n\in\mathcal K_t\setminus\{k\}}C_{n,t}^{\mathrm U}(a)\\
&+\beta_{\mathrm O}C_t^{\mathrm O}(a)
+\beta_{\mathrm R}L_t^{\mathrm{risk}}(a).
\end{aligned}
\label{eq:other_cost_contribution}
\end{equation}
The subscript $-k$ excludes user $k$'s cost contribution,
but does not imply independence from its action. Hence,
\begin{equation}
\phi_t(a)
=\beta_{\mathrm U}C_{k,t}^{\mathrm U}(a)+\phi_{-k,t}(a).
\label{eq:current_potential_decomposition}
\end{equation}
Substituting the local and remote costs gives
\begin{align}
\phi_t(0)
&=\beta_{\mathrm U}C_{k,t}^{\mathrm L}
+\phi_{-k,t}(0),
\label{eq:local_explicit_potential}\\
\phi_t(m)
&=\beta_{\mathrm U}
\left[
\frac{\alpha_{k,t}}{R_{k,m}}
+G_{k,t}^{\mathrm O}(m)
\right]
+\phi_{-k,t}(m).
\label{eq:remote_explicit_potential}
\end{align}

\emph{Case 1: Remote-to-remote, $m\to m'$.}
For distinct $m,m'\in\mathcal M$, subtraction yields
\begin{equation}
\begin{aligned}
\phi_t(m')-\phi_t(m)
={}&\beta_{\mathrm U}\alpha_{k,t}
\left[
\frac{1}{R_{k,m'}}-\frac{1}{R_{k,m}}
\right]\\
&+\beta_{\mathrm U}
\left[
G_{k,t}^{\mathrm O}(m')-G_{k,t}^{\mathrm O}(m)
\right]\\
&+\phi_{-k,t}(m')-\phi_{-k,t}(m).
\end{aligned}
\label{eq:current_potential_case1}
\end{equation}

\emph{Case 2: Local-to-remote, $0\to m$.}
Subtracting \eqref{eq:local_explicit_potential} from
\eqref{eq:remote_explicit_potential} gives
\begin{equation}
\begin{aligned}
\phi_t(m)-\phi_t(0)
={}&\beta_{\mathrm U}
\left[
\frac{\alpha_{k,t}}{R_{k,m}}
+G_{k,t}^{\mathrm O}(m)-C_{k,t}^{\mathrm L}
\right]\\
&+\phi_{-k,t}(m)-\phi_{-k,t}(0).
\end{aligned}
\label{eq:current_potential_case2}
\end{equation}

\emph{Case 3: Remote-to-local, $m\to0$.}
Reversing Case~2 gives
\begin{equation}
\begin{aligned}
\phi_t(0)-\phi_t(m)
={}&\beta_{\mathrm U}
\left[
C_{k,t}^{\mathrm L}
-\frac{\alpha_{k,t}}{R_{k,m}}
-G_{k,t}^{\mathrm O}(m)
\right]\\
&+\phi_{-k,t}(0)-\phi_{-k,t}(m).
\end{aligned}
\label{eq:current_potential_case3}
\end{equation}

These expressions retain the effects on both the selected
user and the remaining system. Since
$\widetilde C_{k,t}=J_t=\phi_t$, every unilateral
change satisfies
\begin{equation}
\widetilde C_{k,t}(a)-\widetilde C_{k,t}(b)
=\phi_t(a)-\phi_t(b),
\label{eq:game_exact_potential_identity}
\end{equation}
establishing the exact-potential property.
A minimizer of $J_t$ exists on finite $\mathcal S_t$
and admits no improving unilateral deviation;
hence, it is a pure Nash equilibrium.
\end{proof}

\subsubsection{Sequential Best Response and Convergence}
\label{subsubsec:best_response_convergence}

The controller visits active users cyclically. At visit $q$,
user $k$'s candidate is
\begin{equation}
a_{k,t}^{(q)}
\in\arg\min_{a\in\mathcal S_{k,t}}
J_t(a,\mathbf s_{-k,t}^{(q)};\mathbf r_t).
\label{eq:asynchronous_best_response}
\end{equation}
It replaces the current action only if
\begin{equation}
\begin{aligned}
&J_t(\mathbf s_t^{(q)};\mathbf r_t)-
J_t(a_{k,t}^{(q)},\mathbf s_{-k,t}^{(q)};\mathbf r_t)
>\epsilon_{\mathrm{BR}},
\end{aligned}
\label{eq:best_response_improvement_tolerance}
\end{equation}
where $\epsilon_{\mathrm{BR}}\ge0$ is the improvement
tolerance; otherwise, the current action is retained.

\begin{theorem}
\label{thm:current_finite_convergence}
Under Proposition~\ref{thm:current_ordinal_potential},
at most $|\mathcal S_t|-1$ changes are accepted.
An unchanged complete sweep certifies an
$\epsilon_{\mathrm{BR}}$-Nash equilibrium, with a pure
Nash equilibrium when $\epsilon_{\mathrm{BR}}=0$.
\end{theorem}

\begin{proof}
Successive accepted profiles satisfy
\[
\phi_t(\mathbf u^{(j+1)};\mathbf r_t)
<
\phi_t(\mathbf u^{(j)};\mathbf r_t)-\epsilon_{\mathrm{BR}}.
\]
Thus, no profile recurs and
$N_{\mathrm{acc}}\le|\mathcal S_t|-1$.
An unchanged sweep verifies
\begin{equation}
\begin{aligned}
J_t(\widehat{\mathbf s}_t;\mathbf r_t)
\le{}&
J_t(a,\widehat{\mathbf s}_{-k,t};\mathbf r_t)
+\epsilon_{\mathrm{BR}},\\
&\forall k\in\mathcal K_t,\quad a\in\mathcal S_{k,t},
\end{aligned}
\label{eq:maximum_unilateral_improvement}
\end{equation}
which is the equilibrium condition since
$\widetilde C_{k,t}=J_t$.
\end{proof}

The result guarantees $\epsilon_{\mathrm{BR}}$-coordinatewise optimality, but not global optimality or convergence of the overarching hierarchical procedure. While an iteration cap ensures termination, it yields no equilibrium certificate. Note that computing-capacity limits influence the common cost via execution and queueing delays, with additional overload penalties applied in the infrastructure experiments.

\subsection{DDQN-Based Adaptive BRF Optimization}
\label{subsec:ddqn_brf_optimization}

The upper layer selects update admission and redundancy
using an aggregate observation $\mathbf x_t$. For each
selected BRF profile, lower-layer offloading determines
the service outcomes used to compute the reward and next
observation. The observation approximates the underlying
system state; exact Markov sufficiency is not assumed.

\subsubsection{Observation}

The observation summarizes demand, update urgency,
configured security requirements, and operating conditions
to guide commitment and redundancy decisions:
\begin{equation}
\begin{aligned}
\mathbf x_t=\Big[
&\mathbf N_t,\,
\{n_{p,t},a_{p,t}^{\mathrm{obs}},d_{p,t}^{\mathrm{obs}},
b_{p,t},m_{p,t},\rho_{p,t}\}_{p\in\mathcal P},\\
&\kappa_t,q_t,e_t
\Big].
\end{aligned}
\label{eq:compact_brf_state}
\end{equation}

Here, $\mathbf N_t=(N_{f,t})_{f\in\{0\}\cup\mathcal F}$,
with $N_{f,t}=K^{-1}\sum_{k\in\mathcal K}
\mathds{1}\{\mu_k^{(t)}=f\}$, gives request fractions
including inactivity ($f=0$).
Per-update features are affected-user fraction
$n_{p,t}=N_p^{(t)}/K$,
normalized age
$a_{p,t}^{\mathrm{obs}}
=\operatorname{clip}_{[0,2]}(A_p^{(t)}/A_p^{\max})$,
magnitude
$d_{p,t}^{\mathrm{obs}}
=\operatorname{clip}_{[0,1]}(D_p^{(t)})$,
fault bound
$b_{p,t}=\bar z_p^{(t)}/\max\{z_{\max},1\}$,
minimum redundancy
$m_{p,t}=r_{p,\min}^{(t)}/r_{\max}$,
and previous redundancy
$\rho_{p,t}=r_p^{(t-1)}/r_{\max}$,
where $z_{\max}$ is the configured fault-bound
normalization scale.

The remaining features are active-user fraction
$\kappa_t=|\mathcal K_t|/K$, normalized uplink quality
\[
q_t=\operatorname{clip}_{[0,1]}
\left[
\frac{1}{12K}\sum_{k\in\mathcal K}
\log_{10}\left(
1+\frac{P_k^{\mathrm{UL}}h_k^{\mathrm{UL}}(t)}
{\sigma_{\mathrm{UL}}^2}
\right)
\right],
\]
and normalized episode progress $e_t$ (elapsed steps
divided by the episode horizon). Throughout,
$\operatorname{clip}_{[a,b]}(x)=\min\{b,\max\{a,x\}\}$.

\subsubsection{Action}

Using the component set \eqref{eq:brf_action_set} and
resource weights defined after $(P1)$, the joint feasible
set is
\begin{equation}
\begin{aligned}
\mathcal A(\mathbf x_t)=
\Bigl\{\mathbf r:\;&r_p\in\mathcal R_p^{(t)},\ \forall p,\\
&\sum_{p\in\mathcal P}w_{p,r_p}\le B_{\mathrm{BC}}
\Bigr\}.
\end{aligned}
\label{eq:compact_joint_brf_action}
\end{equation}
Zero denotes deferral; positive actions provision
security-feasible commitment attempts.

\subsubsection{Reward}

Using the indicators in Section~\ref{subsec:cost_model},
let $\mathcal E_t=\{(p,k):r_p^{(t)}>0,\,
k\in\mathcal K_p^{(t)}\}$ be the attempted delivery set.
Secure-fresh delivery success is
$\eta_t^{\mathrm{ST}}
=\frac{\sum_{(p,k)\in\mathcal E_t}F_{p,k,t}}
{\max\{|\mathcal E_t|,1\}}$.
Across active users, SLA satisfaction and joint
secure-fresh-SLA success are
$\eta_t^{\mathrm{SLA}}
=\frac{\sum_{k\in\mathcal K_t}(1-v_k^{(t)})}
{\max\{|\mathcal K_t|,1\}}$
and
$\eta_t^{\mathrm J}
=\frac{\sum_{k\in\mathcal K_t}Z_{k,t}(1-v_k^{(t)})}
{\max\{|\mathcal K_t|,1\}}$, respectively.

For $C_t=J_t(\widehat{\mathbf s}_t(\mathbf r_t);\mathbf r_t)$,
normalized operating cost and budget utilization are
$\widetilde C_t=\operatorname{clip}_{[0,c_{\max}]}
(C_t/C_{\mathrm{ref}})$ and
$u_t^{\mathrm{BC}}=B_{\mathrm{used}}^{(t)}/B_{\mathrm{BC}}$,
where $C_{\mathrm{ref}}>0$ is fixed.
The deferred and obsolete-deferred fractions are
\begin{equation}
\begin{aligned}
d_t&=
\frac{\sum_{p\in\mathcal P_t^{\mathrm{act}}}
\mathds{1}\{r_p^{(t)}=0\}}
{\max\{|\mathcal P|,1\}},\\
o_t&=
\frac{\sum_{p\in\mathcal P_t^{\mathrm{act}}}
\mathds{1}\{r_p^{(t)}=0,\ A_p^{(t)}\ge A_p^{\max}\}}
{\max\{|\mathcal P|,1\}}.
\end{aligned}
\label{eq:reward_deferral_fractions}
\end{equation}
Both use the total persistent-function population.
The redundancy-margin penalty in
\eqref{eq:reward_modeled_loss} is normalized as
$\widetilde L_t=\operatorname{clip}_{[0,\ell_{\max}]}
\bigl(L_t^{\mathrm{risk}}/\max\{|\mathcal K_t|,1\}\bigr)$.

We group the weighted metrics into service incentives,
resource penalties, and update-management penalties:
\begin{equation}
\begin{aligned}
R_t^{\mathrm{svc}}
&=\vartheta_{\mathrm J}\eta_t^{\mathrm J}
+\vartheta_{\mathrm{ST}}\eta_t^{\mathrm{ST}}
+\vartheta_{\mathrm{SLA}}\eta_t^{\mathrm{SLA}},\\
P_t^{\mathrm{res}}
&=\vartheta_{\mathrm C}\widetilde C_t
+\vartheta_{\mathrm B}u_t^{\mathrm{BC}},\\
P_t^{\mathrm{upd}}
&=\vartheta_{\mathrm D}d_t+\vartheta_{\mathrm O}o_t
+\vartheta_{\mathrm L}\widetilde L_t.
\end{aligned}
\label{eq:reward_components}
\end{equation}
Using weighted scalarization
\cite{Roijers2013MultiObjective,Hayes2022PracticalMORL},
the clipped reward is
\begin{equation}
R_t=\operatorname{clip}_{[R_{\min},R_{\max}]}
\left(R_t^{\mathrm{svc}}-P_t^{\mathrm{res}}
-P_t^{\mathrm{upd}}\right).
\label{eq:compact_brf_reward}
\end{equation}
All weights are nonnegative and fixed during training.
Table~\ref{tab:notation_parameters} reports their values,
normalization scales, and clipping limits.

The upper-layer objective is
\begin{align}
(P3):\quad \max_{\pi}\quad&
\mathbb E_{\pi}\left[\sum_{t=0}^{\infty}\gamma^tR_t\right]
\label{eq:compact_upper_problem}\\
\mathrm{s.t.}\quad&
\mathbf r_t=\pi(\mathbf x_t)\in\mathcal A(\mathbf x_t),
\quad\forall t,
\nonumber
\end{align}
where $0\le\gamma<1$. Auxiliary incentives reweight
effects already included in operating cost. Together
with discounting and clipping, they make $(P3)$ a
surrogate rather than an equivalent reformulation of
$(P1)$. Action selection independently enforces configured
redundancy and budget constraints.

\subsubsection{Branch-Wise DDQN and Feasible Selection}

To address the exponential joint action space, we use
per-update action branches motivated by
\cite{tavakoli2018branching}, with the additive approximation

\begin{equation}
Q(\mathbf x,\mathbf r;\boldsymbol\theta)
\approx\sum_{p\in\mathcal P}
Q_p(\mathbf x,r_p;\boldsymbol\theta),
\label{eq:compact_branch_q}
\end{equation}
using $|\mathcal P|(r_{\max}+1)$ outputs without
assuming separability of the true action-value function.

Insecure positive actions are masked. For a budget
discretization step $\delta_B>0$, define
$\widehat w_{p,r}=\lceil w_{p,r}/\delta_B\rceil$
and $\widehat B=\lfloor B_{\mathrm{BC}}/\delta_B\rfloor$.
Dynamic programming solves the multiple-choice knapsack
problem
\begin{align}
\max_{\mathbf r}\quad&
\sum_{p\in\mathcal P}Q_p(\mathbf x_t,r_p;\boldsymbol\theta)
\label{eq:compact_knapsack}\\
\mathrm{s.t.}\quad&
r_p\in\mathcal R_p^{(t)},\quad \forall p,
\nonumber\\
&\sum_{p\in\mathcal P}\widehat w_{p,r_p}\le\widehat B.
\nonumber
\end{align}
Denote the selected action by
\begin{equation}
\mathbf r_t=
\mathcal K_{\mathrm{BRF}}
\bigl(\mathbf x_t,\mathbf Q(\mathbf x_t;\boldsymbol\theta)\bigr).
\label{eq:compact_knapsack_operator}
\end{equation}

\begin{proposition}[Feasible BRF Selection]
\label{prop:compact_brf_feasibility}
Every action returned by $\mathcal K_{\mathrm{BRF}}$
satisfies the configured redundancy requirements and
blockchain-resource budget.
\end{proposition}

\begin{proof}
Each positive selected action satisfies
$r_p\ge3\bar z_p^{(t)}+1$. By conservative rounding
and the knapsack constraint,
\[
\sum_p w_{p,r_p}
\le \delta_B\sum_p\widehat w_{p,r_p}
\le \delta_B\widehat B
\le B_{\mathrm{BC}}.
\]
\end{proof}

Conservative discretization can exclude feasible actions
but cannot introduce a budget violation.

Following DDQN~\cite{vanhasselt2016double}, the online
network selects the next action, while the target
network evaluates it. For replay transition
$(\mathbf x_t,\mathbf r_t,R_t,\mathbf x_{t+1},
d_t^{\mathrm{term}})$,
\begin{align}
\mathbf r_{t+1}^{\mathrm{sel}}
&=\mathcal K_{\mathrm{BRF}}
\bigl(\mathbf x_{t+1},
\mathbf Q(\mathbf x_{t+1};\boldsymbol\theta)\bigr),
\label{eq:compact_ddqn_selection}\\
Y_t
&=R_t+\gamma(1-d_t^{\mathrm{term}})
\sum_pQ_p\bigl(\mathbf x_{t+1},
r_{p,t+1}^{\mathrm{sel}};\boldsymbol\theta^-\bigr),
\label{eq:compact_ddqn_target}
\end{align}
where $\mathbf r_{t+1}^{\mathrm{sel}}$ is the feasible
action selected for target evaluation,
$r_{p,t+1}^{\mathrm{sel}}$ is its $p$th component,
and $d_t^{\mathrm{term}}$ indicates termination.
The online network minimizes the Huber loss
\begin{equation}
\mathcal L(\boldsymbol\theta)
=
\mathbb E\left[
\ell_{\mathrm H}\left(
Y_t-\sum_pQ_p(\mathbf x_t,r_p^{(t)};\boldsymbol\theta)
\right)\right].
\label{eq:compact_ddqn_loss}
\end{equation}
Exploration uses feasible actions, and $\boldsymbol\theta^-$
is updated periodically. Only true termination sets $d_t^{\mathrm{term}}=1$;
artificial training-horizon truncation retains
bootstrapping~\cite{pardo2018time}.

\subsubsection{BRF-Aware Downlink Power Allocation}
\label{subsubsec:downlink_power_allocation}

For each candidate offloading profile and selected BRF,
we minimize downlink transmit energy subject to a
reserved delivery-time target. The resulting allocation
enters the system-cost evaluation; the power associated
with the returned profile is applied to delivery.
For an output of $L_k^{\mathrm{DL}}$ bits, define
\begin{equation}
T_k^{\mathrm{preDL}}(t)
=
T_k^{\mathrm{UL}}(t)+T_k^q(t)
+T_k^{\mathrm{DAG}}(t)+T_k^{\mathrm{val}}(t),
\label{eq:predownlink_latency}
\end{equation}
with inapplicable terms omitted. The reserved
transmission time is
\begin{equation}
\tau_k^{\mathrm{DL}}(t)
=
\varphi_k[T_k^{\max}-T_k^{\mathrm{preDL}}(t)]^+,
\quad 0<\varphi_k\le1.
\label{eq:downlink_time_budget}
\end{equation}
Shared-update deliveries use their corresponding
payloads and generation-to-delivery validity deadlines.

For fixed bandwidth, channel gain, and interference,
the transmit-energy allocation solves
\begin{align}
(P_{\mathrm{DL}}):\quad
\min_{0\le P\le P_{\max}^{\mathrm{DL}}}\quad&
P\,T_k^{\mathrm{DL}}(P)
\label{eq:downlink_power_problem}\\
\mathrm{s.t.}\quad&
T_k^{\mathrm{DL}}(P)\le\tau_k^{\mathrm{DL}}(t).
\nonumber
\end{align}
For positive payload, bandwidth, and channel gain,
transmit energy is proportional to
$P/\ln(1+a_kP)$, where
$a_k=h_k^{\mathrm{DL}}(t)/
[I_k^{\mathrm{DL}}(t)+\sigma_{\mathrm{DL}}^2]>0$.
This expression increases with $P>0$, since
$\ln(1+x)>x/(1+x)$ for $x>0$.
Thus, whenever feasible, the optimum is
\begin{equation}
P_k^{\mathrm{req}}(t)
=
\frac{I_k^{\mathrm{DL}}(t)+\sigma_{\mathrm{DL}}^2}
{h_k^{\mathrm{DL}}(t)}
\left[
2^{\frac{L_k^{\mathrm{DL}}}
{B_k^{\mathrm{DL}}\tau_k^{\mathrm{DL}}(t)}}-1
\right].
\label{eq:optimal_downlink_power}
\end{equation}
The target is feasible when
$\tau_k^{\mathrm{DL}}(t)>0$ and
$P_k^{\mathrm{req}}(t)\le P_{\max}^{\mathrm{DL}}$.
Otherwise, capped-power delivery is used. The
allocated power and achieved transmit energy are
\begin{align}
P_k^{\mathrm{DL}}(t)
&=\min\{P_k^{\mathrm{req}}(t),P_{\max}^{\mathrm{DL}}\},
\label{eq:allocated_downlink_power}\\
E_k^{\mathrm{DL}}(t)
&=P_k^{\mathrm{DL}}(t)
T_k^{\mathrm{DL}}\bigl(P_k^{\mathrm{DL}}(t)\bigr).
\label{eq:optimal_downlink_energy}
\end{align}
For positive payloads, $P_k^{\mathrm{req}}=+\infty$
when $\tau_k^{\mathrm{DL}}(t)=0$; zero-payload
deliveries require no transmission energy.

Target feasibility differs from the service SLA
in \eqref{eq:sla_violation_indicator}.
Holding other quantities fixed, longer validation
reduces the remaining transmission time and increases
required power, coupling BRF selection to service cost.

\begin{algorithm}[t]
\caption{DDQN-RPBFT Training}
\label{alg:ddqn_rpbft}
\small
\begin{algorithmic}[1]
\State Initialize $\boldsymbol\theta$,
$\boldsymbol\theta^-=\boldsymbol\theta$, and replay buffer
$\mathcal D$.
\For{each training episode}
    \State Reset the environment and obtain $\mathbf x_1$.
    \For{each slot $t$}
        \State Select $\mathbf r_t$ by feasible exploration
        or \eqref{eq:compact_knapsack_operator}.
        \State Obtain $\widehat{\mathbf s}_t(\mathbf r_t)$
        by sequential strict-improvement updates,
        recomputing downlink power for each candidate.
        \State Execute the decisions and observe
        $R_t$, $\mathbf x_{t+1}$, and $d_t^{\mathrm{term}}$.
        \State Store
        $(\mathbf x_t,\mathbf r_t,R_t,\mathbf x_{t+1},
        d_t^{\mathrm{term}})$ in $\mathcal D$.
        \If{a replay mini-batch is available}
            \State Form targets using
            \eqref{eq:compact_ddqn_selection}--%
            \eqref{eq:compact_ddqn_target}.
            \State Update $\boldsymbol\theta$ using
            \eqref{eq:compact_ddqn_loss}.
        \EndIf
        \State Periodically set
        $\boldsymbol\theta^-\leftarrow\boldsymbol\theta$.
        \If{the episode ends}
            \State \textbf{break}
        \EndIf
    \EndFor
\EndFor
\State \textbf{Return:} learned BRF policy.
\end{algorithmic}
\end{algorithm}

\subsubsection{Computational Complexity}
\label{subsubsec:computational_complexity}

Let $P=|\mathcal P|$, $A=r_{\max}+1$, and
$\widehat B=\lfloor B_{\mathrm{BC}}/\delta_B\rfloor$.
Let $C_{\mathrm{NN}}$ and $C_{\mathrm{eval}}$ denote
the costs of a network forward pass and a candidate-profile
evaluation, including interference, DAG processing,
consensus, and downlink allocation.

Knapsack selection costs $\mathcal O(PA(\widehat B+1))$.
Each offloading sweep evaluates at most $K(M+1)$ profiles;
with at most $I_{\max}$ sweeps, per-slot inference costs
\begin{equation}
\mathcal O\!\left(
C_{\mathrm{NN}}+PA(\widehat B+1)
+I_{\max}K(M+1)C_{\mathrm{eval}}
\right).
\label{eq:algorithm_complexity}
\end{equation}
The network outputs $PA$ scores rather than enumerating
$A^P$ joint actions. Budget discretization yields a
pseudopolynomial DP bound; finer resolution increases
cost. The sweep cap bounds runtime but does not
certify equilibrium.

\section{Numerical Simulation}
\label{sec:numerical}

Table~\ref{tab:notation_parameters} summarizes the default
system, learning, and reward parameters. Experiment-specific
variations are stated with the corresponding results.

\begin{table}[t]
\centering
\caption{Simulation Parameters and Default Values}
\label{tab:notation_parameters}
\renewcommand{\arraystretch}{1.02}
\setlength{\tabcolsep}{3pt}
\footnotesize
\begin{tabularx}{\columnwidth}{l X l}
\hline
\textbf{Symbol} & \textbf{Definition} & \textbf{Value}\\
\hline
\multicolumn{3}{l}{\textit{System parameters}}\\
$K$ & XR users & $25$\\
$\lambda$ & Task-arrival probability & $0.5$\\
$\bar z$ & Byzantine-fault bound & $2$\\
$r_{\min}$ & Minimum feasible BRF & $3\bar z+1$\\
$r_{\max}$ & Maximum BRF & $15$\\
$B_{\mathrm{BC}}$ & Blockchain resource budget & $400$ Mbit\\
$H_f$ & Shared-state validity horizon & $3$ slots\\
$\Delta t$ & Slot duration & $1$ s\\
$T^{\max}$ & Service deadline & $3$ s\\
$F_k^{\mathrm L,\max}$ & Local computing capacity & $2.5$ GHz\\
$P_k^{\mathrm{UL}}$ & Uplink power & $0.2$ W\\
$B^{\mathrm{UL}}$ & Uplink bandwidth & $40$ MHz\\
$B^{\mathrm{DL}}$ & Downlink bandwidth & $20$ MHz\\
$P_{\max}^{\mathrm{DL}}$ & Maximum downlink power & $2$ W\\
$M$ & Uplink channels & $4$\\
$|\mathcal F|$ & XR task types & $4$\\
$F^{\mathrm{FIN}}$ & FIN capacity & $40$ Gcycles/s\\
$F^{\mathrm{EIN}}$ & EIN capacity & $80$ Gcycles/s\\
\hline
\multicolumn{3}{l}{\textit{Cost parameters}}\\
$w_T,w_E,w_A$
& Latency, energy, and age weights & $1,1,0.25$\\
$w_{\mathrm C},w_{\mathrm{DL}}$
& Computation and downlink weights & $1,1$\\
$w_{\mathrm{BC}},w_{\mathrm L},w_{\mathrm V}$
& Validation, margin, and failure weights & $1,1,25$\\
$w_{a,1},w_{a,2}$
& Linear and quadratic age weights & $5,25$\\
$w_d,w_n,w_i,w_o$
& Magnitude, sharing, interaction, obsolescence
& $8,3,20,75$\\
$w_{\mathrm{def}},\lambda_p$
& Deferral-loss scale and sensitivity & $12,4$\\
$\beta_{\mathrm U},\beta_{\mathrm O},\beta_{\mathrm R}$
& System-cost weights & $1,1,1$\\
\hline
\multicolumn{3}{l}{\textit{Learning parameters}}\\
$\eta$ & Learning rate & $10^{-4}$\\
$\gamma$ & Discount factor & $0.95$\\
$|\mathcal B|$ & Mini-batch size & $64$\\
$N_{\mathrm{buf}}$ & Replay capacity & $50{,}000$\\
$\epsilon$ & Exploration probability & $1.0\rightarrow0.05$\\
$N_{\mathrm{target}}$ & Target-update interval & $100$ steps\\
-- & Hidden-layer widths & $256$--$256$--$128$\\
$N_{\mathrm{ep}}$ & Training episodes & $500$\\
$T_{\mathrm{ep}}$ & Steps per episode & $50$\\
-- & Random seeds & $\{11,22,33\}$\\
\hline
\multicolumn{3}{l}{\textit{Reward parameters}}\\
$\vartheta_{\mathrm J}$ & Joint-service weight & $2.0$\\
$\vartheta_{\mathrm{ST}},\vartheta_{\mathrm{SLA}}$
& Marginal-service weights & $0.75,\ 0.75$\\
$\vartheta_{\mathrm C},\vartheta_{\mathrm B}$
& Cost and budget weights & $0.25,\ 0.25$\\
$\vartheta_{\mathrm D},\vartheta_{\mathrm O}$
& Deferral and obsolescence weights & $0.75,\ 0.50$\\
$\vartheta_{\mathrm L}$ & Redundancy-margin penalty weight & $0.25$\\
$C_{\mathrm{ref}}$ & Cost normalization scale & $100$\\
$c_{\max},\ell_{\max}$ & Normalized cost and loss caps & $20,\ 10$\\
$[R_{\min},R_{\max}]$ & Reward clipping interval & $[-10,10]$\\
\hline
\end{tabularx}
\end{table}

\subsection{Benchmarks and Ablation Variants}

We compare DDQN-RPBFT with controlled baselines and
ablations motivated by the cited studies.

\subsubsection{Controlled BRF Baselines}
\begin{itemize}
    \item \textbf{Minimum Secure} uses
    $r_p^{(t)}=3\bar z_p^{(t)}+1$ whenever feasible,
    representing minimum-overhead secure provisioning
    \cite{du2025secure}.

    \item \textbf{Excess Redundancy} selects the largest
    feasible secure BRF to contrast conservative
    provisioning with adaptive BFT
    \cite{leduc2025draftee}.

    \item \textbf{Random Feasible} samples a security-
    and budget-feasible BRF vector, motivated by random
    resource-allocation comparisons
    \cite{yu2022asynchronous}.

    \item \textbf{No Commit} fixes $\mathbf r_t=\mathbf 0$
    as a diagnostic reference without new commitments.
\end{itemize}

\subsubsection{Architectural Ablations}
\begin{itemize}
    \item \textbf{Without Feasibility Filtering (w/o-FF)} 
    removes learner-side security and budget filtering;
    an external safety mechanism handles infeasible
    proposals before execution.

    \item \textbf{Without Voluntary Deferral (w/o-VD)}
    disables voluntary deferral while retaining
    resource-forced deferral, isolating adaptive
    admission relative to secure-offloading
    operation~\cite{du2025secure}.
\end{itemize}

\subsubsection{Learning-Sensitivity Variants}
\begin{itemize}
    \item \textbf{Without Byzantine-State Information (w/o-BS)}
    removes Byzantine-related network inputs while
    retaining feasibility filtering, examining
    threat-state awareness~\cite{leduc2025draftee}.

    \item \textbf{Vanilla DQN} replaces DDQN with
    DQN~\cite{mnih2015human} under otherwise identical
    settings, providing a controlled learning-model
    comparison~\cite{yu2022asynchronous}.
\end{itemize}

\subsection{Convergence Analysis}
\begin{figure}[t]
    \centering
    \includegraphics[width=\columnwidth]{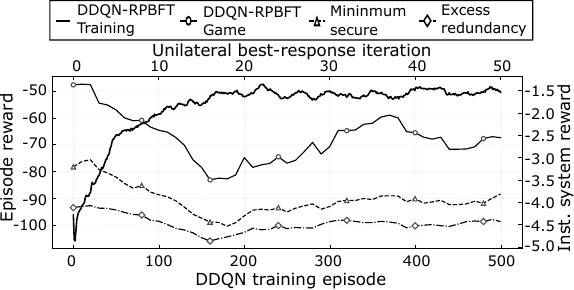}
    \caption{Convergence behavior of DDQN-RPBFT training and the lower-layer best-response game for $K=20$ and $H_f=3$.}
    \label{fig:ddqn_game_convergence}
\end{figure}
Fig.~\ref{fig:ddqn_game_convergence} shows the convergence behavior of DDQN-RPBFT for $K=20$ and $H_f=3$. The episodic return improves rapidly early in training and stabilizes near $-50$, indicating that the learned policy has reached a stable regime. In the lower-layer game, DDQN-RPBFT consistently achieves a higher instantaneous system reward than the Minimum Secure and Excess Redundancy baselines, particularly after the initial best-response iterations. This advantage arises from state-dependent BRF adaptation, which balances redundancy overhead against offloading performance more effectively than fixed-redundancy strategies. These results support both the training stability of DDQN-RPBFT and the benefit of adaptive redundancy during best-response optimization.

\subsection{Blockchain-Oriented Evaluation}

\textit{1) Byzantine Robustness and System Reward:} For policy $\pi$ and baseline $b$, relative reward
improvement is
\[
G_R(\pi;b)
=
100\frac{\bar R_\pi-\bar R_b}{|\bar R_b|},
\qquad \bar R_b\neq0,
\]
where the means use the same user loads. 

Fig.~\ref{fig:main_performance}(a) evaluates the impact of Byzantine
severity on system performance. DDQN-RPBFT consistently outperforms the
fixed and random redundancy schemes as the number of XR users increases.
Using the best-performing baseline, Minimum Secure, as a common
reference, DDQN-RPBFT improves the average reward by $46.7\%$, $33.2\%$,
and $32.4\%$ for $z=1$, $2$, and $3$, respectively. Although the reward
decreases with increasing $z$ and user load, the advantage of DDQN-RPBFT
is maintained and becomes more pronounced at higher loads. This is
because increasing Byzantine severity raises the minimum redundancy
requirement and its associated resource cost, whereas DDQN-RPBFT adapts
redundancy to the prevailing system state rather than relying on a fixed
redundancy level.

\textit{2) Secure timely-update performance:} Next, we evaluate secure timely-update performance through
\[
\bar\eta_\pi^{\mathrm{ST}}
=
\frac{1}{T}\sum_{t=1}^{T}\eta_{\pi,t}^{\mathrm{ST}},
\]
where $\eta_{\pi,t}^{\mathrm{ST}}$ is the fraction of
attempted deliveries satisfying the security,
delivery-feasibility, and freshness requirements.
Results are subsequently averaged across seeds. Thus, $\bar\eta_\pi^{\mathrm{ST}}$ measures the average
secure timely-delivery success among attempted
user-update deliveries under policy $\pi$.

\begin{figure*}[!t]
    \centering
    \includegraphics[
    width=0.9\textwidth,
    height=0.5\textheight,
    keepaspectratio]{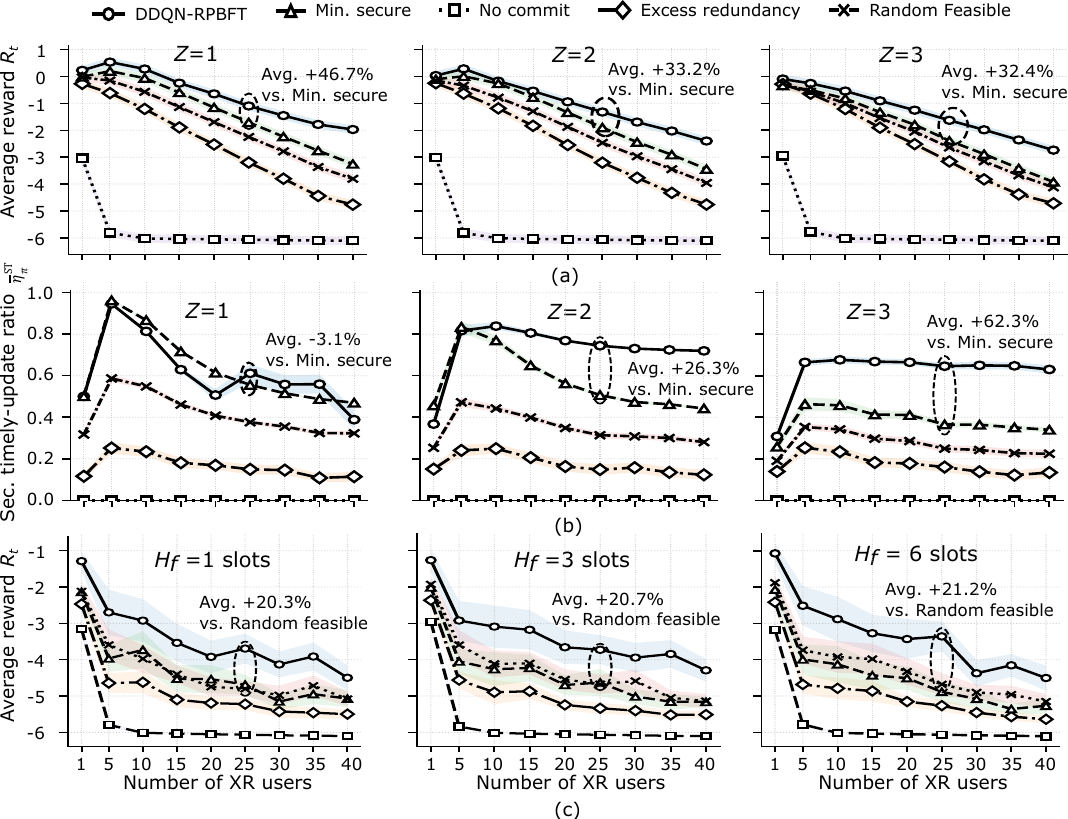}
    \caption{Performance of DDQN-RPBFT under security and shared-state
    validity variations: (a) Impact of Byzantine severity on average system reward (b) secure timely-update ratio under different
    Byzantine levels; and (c) Impact of shared-state validity horizon on average system reward. Shaded regions denote 95\%
    confidence intervals where shown.}
    \label{fig:main_performance}
\end{figure*}
Fig.~\ref{fig:main_performance}(b) demonstrates that adaptive redundancy becomes increasingly beneficial as Byzantine severity rises. Compared to Minimum Secure, DDQN-RPBFT changes the average secure timely-update ratio by $-3.1\%$, $+26.3\%$, and $+62.3\%$ for $z=1$, $2$, and $3$, respectively. At $z=1$, Minimum Secure performs slightly better due to lower PBFT requirements and minimal redundancy overhead. However, DDQN-RPBFT outperforms at higher user counts ($K=25$--$35$) through adaptive commit/deferral and BRF selection. As $z$ increases, fixed secure commitment becomes more costly, while DDQN-RPBFT efficiently coordinates commitment and redundancy to maintain a higher fraction of timely secure deliveries. The gain from $26.3\%$ at $z=2$ to $62.3\%$ at $z=3$ underscores the advantage of adaptation in stringent Byzantine conditions.

\subsection{User-Oriented Evaluation}

We first examine the effect of the shared-state validity horizon
$H_f$, which defines the number of slots an admitted
persistent update remains valid before a new secure update is needed.
A smaller $H_f$ leads to more frequent secure updates, while a
larger $H_f$ allows longer reuse of committed state.

Fig.~\ref{fig:main_performance}(c) shows DDQN-RPBFT consistently achieves the highest average reward as XR user count increases. Compared to the strongest baseline, Random Feasible, it yields average gains of $20.3\%$, $20.7\%$, and $21.2\%$ for $H_f=1$, $3$, and $6$. Reward declines with higher $K$ due to increased contention for resources, but DDQN-RPBFT's persistent advantage highlights its joint adaptation of update admission and redundancy across varying state-validity requirements.

\subsection{Infrastructure-Oriented Evaluation}

For execution location $\ell_{k,t}$ assigned to user $k$,
tier utilization is
$\rho_{\ell,t}
=\sum_{k\in\mathcal K_t:\ell_{k,t}=\ell}
C_{\mu_k^{(t)}}/(F_\ell^{\max}\Delta t)$,
where $\ell\in\{\mathrm{FIN},\mathrm{EIN}\}$.
In these experiments, the operator cost additionally
includes $w_q\sum_\ell T_{\ell,t}^{q}
+w_\rho\sum_\ell[\rho_{\ell,t}-1]^+$,
where $T_{\ell,t}^{q}$ is the modeled queueing delay,
and $w_q,w_\rho\ge0$ weight queueing delay and
excess utilization, respectively.
Saturation or overload is measured by the empirical
frequency of $\rho_{\ell,t}\ge1$.

The remote-task admission ratio is
\[
a_t^{\mathrm{adm}}
=
\frac{\sum_{k\in\mathcal K_t}
\mathds{1}\{s_k^{(t)}>0\}}
{|\mathcal K_t|},
\qquad |\mathcal K_t|>0.
\]
Admission and placement shares are averaged over
slots containing active users.
Fig.~\ref{fig:infrastructure_performance}(a)--(c)
examines adaptation to increasing task-arrival intensity.
The admission ratio decreases with $\lambda$ and $K$.
For $K=40$, increasing $\lambda$ from $0.2$ to $0.9$
reduces admission by approximately $67\%$.
FIN and EIN saturation-or-overload frequencies remain
negligible in most cases, increasing under the heaviest
tested load $(\lambda=0.9,K=40)$. These results
characterize adaptation to resource contention without
establishing queue stability.

For $\ell\in\{\mathrm L,\mathrm{FIN},\mathrm{EIN}\}$,
the placement share is
\[
s_t^\ell
=
\frac{\sum_{k\in\mathcal K_t}
\mathds{1}\{\ell_{k,t}=\ell\}}
{|\mathcal K_t|},
\qquad |\mathcal K_t|>0,
\]
averaged over slots containing active users.
Fig.~\ref{fig:infrastructure_performance}(d)--(e)
evaluates FIN/EIN computing capacities
$(F_{\mathrm{FIN}}^{\max},F_{\mathrm{EIN}}^{\max})
=(20,40)$ and $(80,160)$ Gcycles/s, respectively.
As $K$ increases, DDQN-RPBFT shifts more execution
to local processing. At $K=40$, higher capacities
reduce the local-placement share from $68\%$ to $61\%$,
corresponding to approximately $22\%$ more combined
FIN/EIN placement. Capacity-normalized system cost
increases with $K$, while normalized BRF remains
stable, indicating placement adaptation to available
computing resources.

\begin{figure*}[!t]
    \centering
    \includegraphics[width=\textwidth]
    {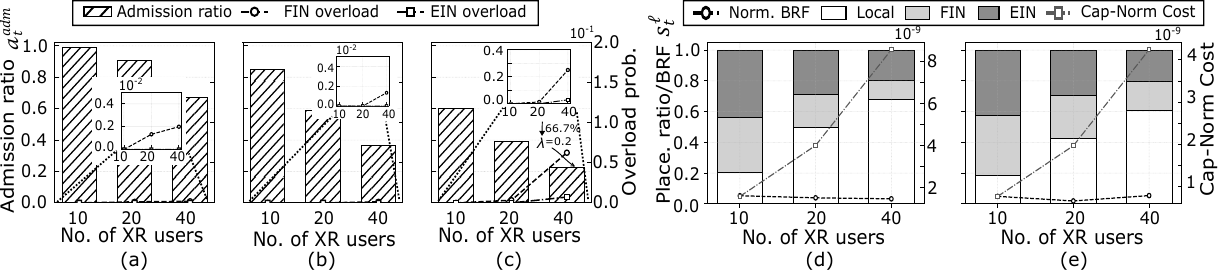}
    \caption{Infrastructure adaptation of DDQN-RPBFT:
(a)--(c) remote-task admission ratio and FIN/EIN
saturation-or-overload frequency under task-arrival
probabilities $\lambda=0.2$, $0.5$, and $0.9$,
respectively; and (d)--(e) computation placement,
normalized BRF, and capacity-normalized system cost
for $(F_{\mathrm{FIN}}^{\max},F_{\mathrm{EIN}}^{\max})
=(20,40)$ and $(80,160)$ Gcycles/s, respectively.}
    \label{fig:infrastructure_performance}
\end{figure*}

\subsection{Ablation and Sensitivity Analysis}

For each matched run at a given $(z,K)$ setting,
performance is averaged over the final 100 episodes.
Methods are compared using the Friedman
test~\cite{demsar2006statistical}; lower average
ranks indicate better performance.

Fig.~\ref{fig:d2_ablation}(a) shows that w/o-FF
and w/o-VD perform worse, supporting the benefits
of feasible-action filtering and adaptive deferral. Fig.~\ref{fig:d2_ablation}(b) shows that w/o-BS
achieves the lowest average ranks, while DDQN-RPBFT
and Vanilla DQN rank similarly. These results
suggest limited benefit from explicit Byzantine-state
inputs under the tested stationary conditions.
Adaptation to time-varying Byzantine conditions
remains for future work.

\begin{figure}[!t]
    \centering
    \includegraphics[width=\columnwidth]
    {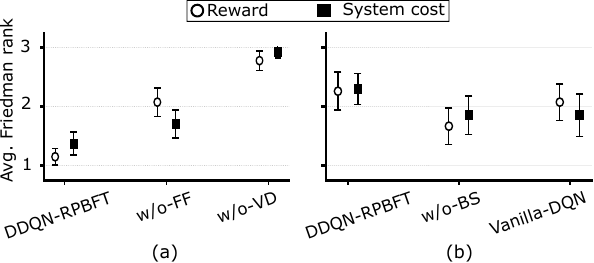}
    \caption{Friedman average-rank analysis of DDQN-RPBFT:
    (a) main architectural ablations and (b) learning-sensitivity
    variants. Lower average rank indicates better relative performance
    for the corresponding metric; error bars denote 95\% confidence
    intervals.}
    \label{fig:d2_ablation}
\end{figure}
\section{Related Work}
\label{sec:related_work}

Table~\ref{tab:related_comparison} summarizes the related
research directions. Our work relates to XR offloading,
where heterogeneous functions execute across user and
edge resources~\cite{alriksson2021xr}.
Yu \emph{et al.}~\cite{yu2024dependent} optimize
dependency-constrained XR offloading, while
Lia \emph{et al.}~\cite{lia2022in} address
delay-constrained in-network task placement.
Our earlier work~\cite{aliyu2023dynamic} combines game
theory and DDQN for partial offloading in COIN.
The present work additionally considers persistent-output
commitment and its competition with application execution.

Related to ours are multi-user XR architectures supporting
shared mapping, rendering, and state synchronization
\cite{chen2023networked,huang2024scaxr}.
Slocum \emph{et al.}~\cite{slocum2024sharedstate}
demonstrate how shared-state manipulation affects other
users' environments. These studies motivate protection
but do not jointly address Byzantine redundancy and
offloading under shared resource constraints.
We restrict commitment to persistent authoritative
updates, excluding transient intermediate outputs.

Our work also relates to BFT consensus. Byzantine
agreement and PBFT establish fault-tolerance conditions
under their protocol assumptions
\cite{lamport2019byzantine,castro1999pbft}.
LA-BFT uses reputation-based adaptive consensus
\cite{ye2024labft}, Stable BFT addresses unreliable
links~\cite{ling2024stable}, and DRAFTEE adapts validator
participation to throughput and security requirements
\cite{leduc2025draftee}. Kuznetsov and Tonkikh
\cite{kuznetsov2022reconfiguration} preserve consistency
during Byzantine reconfiguration. Assuming certified
configuration transitions, we examine how application
demand and update significance guide admission and
redundancy provisioning.

The closest studies coordinate blockchain and edge
computing resources. Fan \emph{et al.}
\cite{fan2024resource} match computation and consensus
resources; Fang \emph{et al.}~\cite{fang2024task}
optimize offloading and resource allocation.
Du \emph{et al.}~\cite{du2025secure} incorporate
improved PBFT into secure offloading, while
Chen \emph{et al.}~\cite{chen2025blockchain} jointly
optimize offloading and blockchain processing in
multi-UAV MEC. Our focus is the coupling of
persistent-update admission, semantic significance,
and security-feasible redundancy with XR execution
under shared resource constraints.

Finally, Chen \emph{et al.}~\cite{chen2015efficient}
model multiuser offloading under interference as a
potential game; our identical-interest game evaluates
unilateral changes through aggregate system cost.
Shielding enforces specified safety properties during
learning~\cite{alshiekh2018shielding}, and invalid-action
masking has been studied for policy-gradient methods
\cite{huang2022masking}. Our DDQN combines per-update
security masks with joint budget-constrained selection,
enforcing feasibility independently of learned rewards.

\section{Conclusion}
\label{sec:conclusion}

This paper presented semantic-aware update admission and adaptive Byzantine redundancy for secure computation offloading in COIN-enabled metaverse systems. The proposed hierarchical R-PBFT approach addresses
sequential update admission and redundancy control
using feasibility-filtered DDQN, and conditional
slot-wise offloading through an identical-interest game. Results show improved system performance, secure timely-update delivery, and infrastructure adaptation under diverse loads and Byzantine severity. Ablation studies confirm the value of feasible-action filtering and resource-aware deferral. Future work will address Byzantine-state-aware adaptation in non-stationary adversarial settings.

\section*{Acknowledgment}
The authors thank Sélinde Van Engelenburg, Cho-Rong Yu,
and Tai-Won Um for their valuable feedback and contributions
during the conceptualization and early development of this work.

\bibliographystyle{IEEEtran}
\bibliography{main}

\vspace{-0.5\baselineskip}
\begin{IEEEbiography}[{\includegraphics[width=0.8in,height=1in,clip,keepaspectratio]{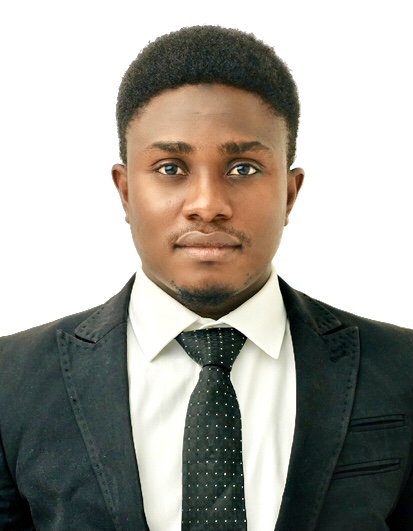}}]{Ibrahim Aliyu}
(Member, IEEE) received the B.Eng. and M.Eng. degrees in computer engineering from the Federal University of Technology, Minna, Nigeria, in 2014 and 2018, respectively, and the Ph.D. degree in computer science and engineering from Chonnam National University, Republic of Korea, in 2022. He is currently a Research Professor with the Intelligent Innovation G5-AICT Research Center, Chonnam National University, Republic of Korea. His research focuses on distributed and in-network computing, edge intelligence, digital twins, and resource orchestration for next-generation networks and cyber-physical systems. His research interests also include computation offloading, trustworthy artificial intelligence, federated learning, network security, and autonomous network management.
\end{IEEEbiography}

\vspace{-0.5\baselineskip}

\begin{IEEEbiography}[{\includegraphics[width=0.8in,height=1in,clip,keepaspectratio]{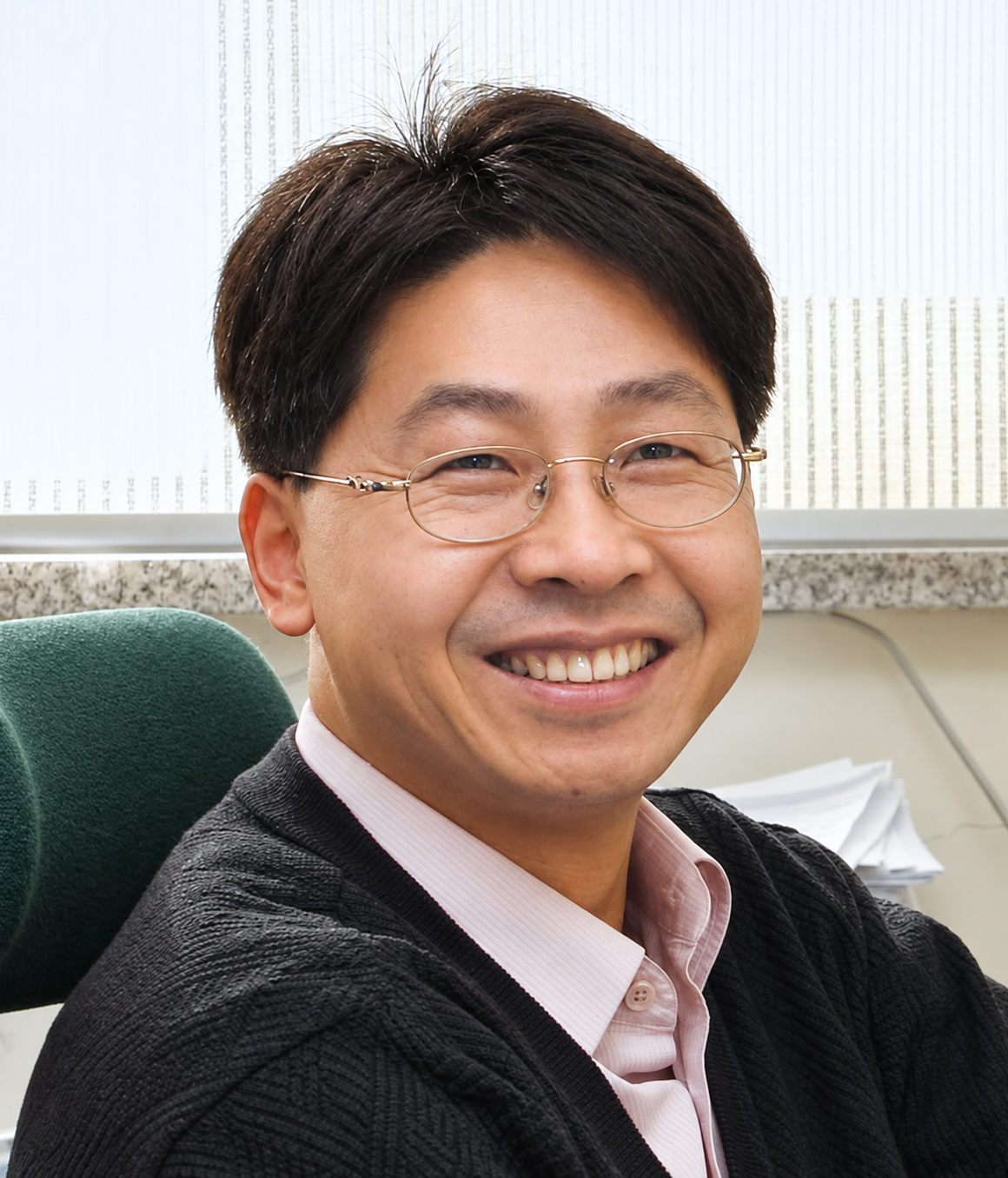}}]{Jaehyung Park}
received the B.S. degree in computer science from Yonsei
University, South Korea, in 1991, and the M.S. and Ph.D.
degrees in computer science from the Korea Advanced
Institute of Science and Technology (KAIST), South Korea,
in 1993 and 1997, respectively. He was with the Center
for Artificial Intelligence Research at KAIST from 1997
to 1998 and the Network Laboratory at ETRI from 1998
to 2002. Since 2002, he has been with Chonnam National
University, South Korea, where he is currently a
professor with the School of Electronics
and Computer Engineering. His research interests include
anomaly detection, attention and gating mechanisms,
digital twins, edge and mobile edge computing,
functional networks, heterogeneous graphs and networks,
and the Internet of Things.
\end{IEEEbiography}

\begin{IEEEbiography}[{\includegraphics[width=0.8in,height=1.in,clip,keepaspectratio]{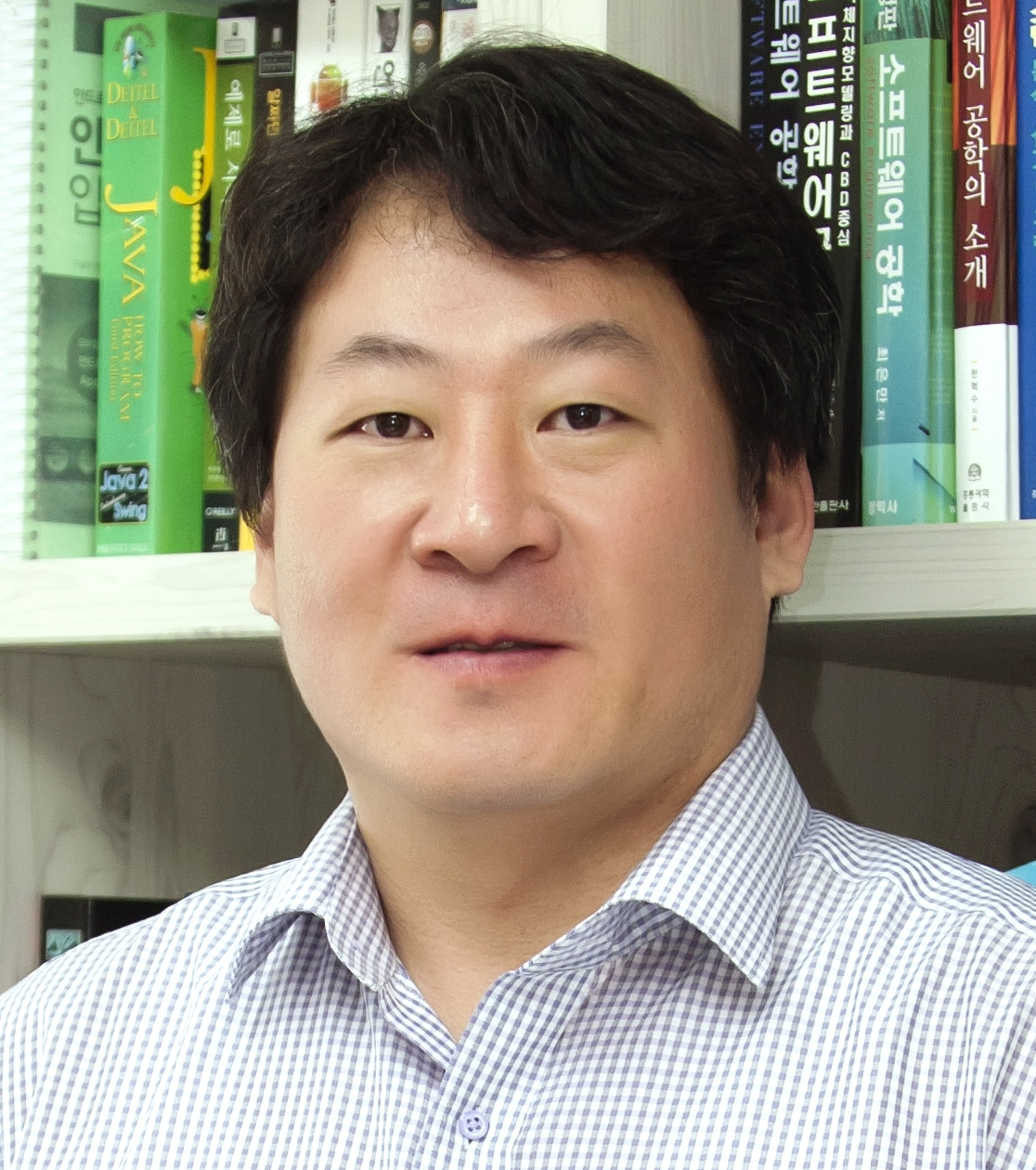}}]{Jinsul Kim}
(Member, IEEE) received a BS degree in computer science from the University of Utah at Salt Lake City in Utah, USA, in 1998 and MS (2004) and PhD (2008) degrees in digital media engineering from the Korea Advanced Institute of Science and Technology (KAIST) in Daejeon, South Korea. Previously, he worked as a researcher at the Broadcasting/Telecommunications Convergence Research Division, Electronics and Telecommunications Research Institute (ETRI) in Daejeon, Korea, from 2004 to 2009 and was a professor at Korea Nazarene University in Cheonan, Korea, from 2009 to 2011. He is a professor at Chonnam National University, Gwangju, Korea. He is a member of the Korean national delegation for ITU-T SG13 international standardization. He has participated in various national research projects and domestic and international standardization activities. He is a co-research director of the Artificial Intelligence Innovation Hub Research and Development Project hosted by Korea University and is the director of the G5-AICT research center.

\end{IEEEbiography}

\end{document}